\documentclass[journal]{IEEEtran}

\usepackage[utf8]{inputenc}
\usepackage[T1]{fontenc}
\usepackage{microtype}
\usepackage{booktabs}
\usepackage{array}
\usepackage{xcolor}
\usepackage{amsmath}
\usepackage{amssymb}
\usepackage{graphicx}
\usepackage[hidelinks]{hyperref}
\hypersetup{pdftitle={Sensitive-Topic Leakage Through LLM Routing Metadata: Measurement and Mitigation}, pdfauthor={Teng-Ruei Chen}}
\usepackage{xurl}
\usepackage{needspace}

\newcommand{\code}[1]{\texttt{#1}}
\newcommand{\dadj}{\Delta_{\mathrm{adj}}}
\newtheorem{proposition}{Proposition}
\newtheorem{remark}{Remark}
\title{Sensitive-Topic Leakage Through LLM Routing Metadata: Measurement and Mitigation}

\author{Teng-Ruei~Chen%
\thanks{T.-R. Chen is with Krixvon, Taipei 100, Taiwan (e-mail: \mbox{ymchen.bi04g@g2.nctu.edu.tw}; ORCID: 0000-0001-8995-334X).}%
\thanks{Preprint, October 2026. Every study in Table~\ref{tab:prereg} was pre-registered: its protocol was committed to our private
repository before its data were obtained or, if it reused data, before its new step. E-3 to E-5 re-analyzed WildChat shards whose
earlier results were known; E-8 was frozen after E-5b had measured the two stickiness defenses post hoc. Protocols, code, and
aggregate results are released (Artifact Statement).}}
\begin{document}
\maketitle

\begin{abstract}
LLM routers pick a cheap or expensive model per request by its content, and many gateways and some cloud platforms can log that choice with content logging off. We measure this privacy channel beyond token counts, accounting for noisy labels and repeated prompts. We run pre-registered studies on 1.7 million real requests (WildChat-1M, LMSYS-Chat-1M) with two cost/quality routers and a domain router, survey eleven systems' logging, and test post-processing defenses. At matched length, the shift's direction depends on category and router. For RouteLLM at the 50\% operating point, harassment and self-harm requests reach the strong model 19 points less often than comparable ones on prompts unseen in exploration, medical requests (exploratory: LLM labels failed their gate) 31 points less often on distinct prompts (both post hoc), and sexual requests 10 points more often (secondary); the other router's four are negative. Twenty RouteLLM decisions separate frequent medical askers with AUC 0.71, exploratory and below the pre-registered primary endpoint's 0.75 (domain router: 0.92, an upper estimate). Per-category length-matched parity with accurate labels removes the gap on real traffic, costing at most 0.2 accuracy points on RouterBench (post hoc), where routers' gaps on sensitive subjects (13--42 points, pre-registered) exceed those of an oracle routing by realized accuracy gain (1--11, post hoc). Per-conversation stickiness, per-user budget bands, and pooled parity fail, the last as categories' shifts differ in size or sign. A post hoc exact per-user rate hides only even-prefix strong counts and forfeits most self-assessed routing value; it preserves odd-position decisions, from which a post hoc log attack reaches AUC 0.73 after 20 RouteLLM requests (exploratory).
\end{abstract}


\begin{figure*}[!t]
\centering
\includegraphics[width=\textwidth]{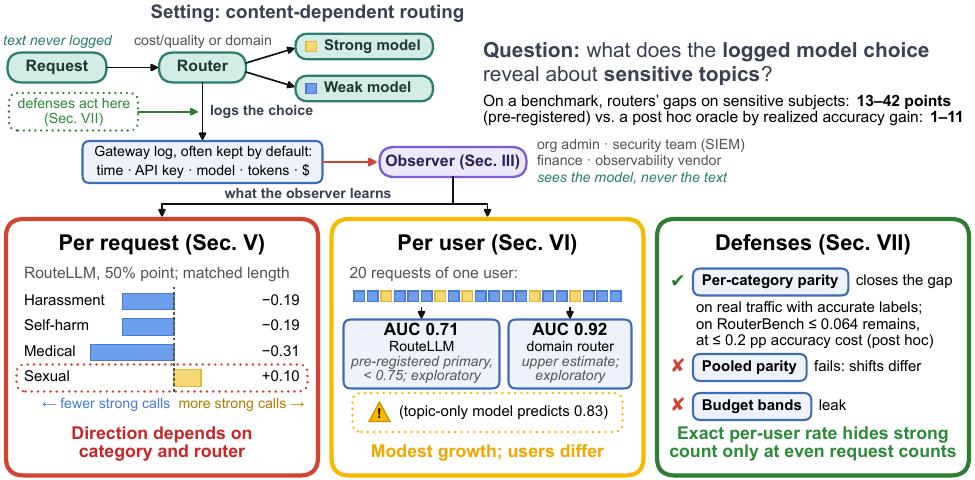}
\caption{Overview. Top: gateways can keep the router's choice, the caller's key, and token counts, often by default and with content
logging off (Section~\ref{sec:threat}). Bottom: what these choices reveal per request (Section~\ref{sec:perreq}; RouteLLM's gaps:
harassment and self-harm on never-seen and medical on distinct WildChat prompts are post hoc, medical labels exploratory, sexual
secondary; notdiamond-0001's four gaps are negative) and per user (Section~\ref{sec:userlevel}; 0.71 is the pre-registered primary
endpoint, below 0.75; both AUCs are exploratory as the label gate failed), and what defenses cost (Section~\ref{sec:defense}; post hoc: parity's RouterBench cost and
the exact rate, under which original A2 reaches 0.53 and an odd-position log attack reaches 0.73 after 20 requests). Benchmark line: RouterBench; post hoc oracle
at the same budget.}
\label{fig:overview}
\end{figure*}

\section{Introduction}

Large language model (LLM) applications increasingly place a routing or cascade layer, which we call a \emph{router}, in front of several
models. For each request it either picks one model directly or queries a cheap model first and escalates to a stronger one when needed,
trading answer quality against cost~\cite{chen2023frugalgpt,ding2024hybrid,aggarwal2023automix,ong2024routellm,hu2024routerbench}. Gateways and cloud platforms now offer routers as built-in options that a caller
selects by model name~\cite{openrouterauto}. A router's decision is, by construction, a function of the content of the request.

Prior work falls into three lines, reviewed in full in Section~\ref{sec:related}. Routing research comprises cascades and learned routers that trade cost against quality~\cite{chen2023frugalgpt,aggarwal2023automix,ding2024hybrid,ong2024routellm}, benchmarks for them~\cite{hu2024routerbench}, analyses showing that their decisions follow the query's category and register~\cite{kassem2025robust,koul2026register}, and domain routers that pick a model by the request's subject~\cite{vllmsr}. Privacy-preserving routing keeps sensitive content from remote models or from the router: a local model protects what is sent to a proprietary one~\cite{papillon2025}, requests are routed between edge and cloud by their sensitivity~\cite{prism_router}, and secure computation hides the request from the router~\cite{pproute2026}. Side-channel research shows that observable behavior that depends on private inputs reveals them: location-aware path selection in Tor reveals where clients are~\cite{wails2018tempest,rochet2020claps}, the exits and timing of adaptive networks reveal predictions and input attributes~\cite{multiexit2023ccs,akinsanya2024ndss}, the packet sizes and timing of encrypted LLM traffic reveal prompt topics and the serving model~\cite{whisperleak2025,pouryousef2026traffic}, and, in agent systems, which agents are contacted and which messages are chosen reveal task types, private state, and intent~\cite{dangol2026graph,xu2026sicc,peng2026agentic}.

None of these lines examines a model router's decision as it is logged. Many gateways, and some cloud platforms, can record by design the served model and the caller of each request, often with token counts, even when content logging is off, and in some of them this is the default; some routing products also write the decision itself as a named field (Section~\ref{sec:threat}). Privacy-preserving routing systems protect content but do not analyze or defend the post-routing operational record considered here; when the route depends on sensitivity, the record carries the sensitivity label itself. The only proposal we found that randomizes the local/remote delegation bit, and bounds what it leaks, is evaluated only in simulation~\cite{gaikwad2025avec}. Agent-level work shows that selections leak~\cite{dangol2026graph,xu2026sicc,peng2026agentic,kim2026agentcardleak}, but there agents make the selection, and observers read it from network traffic, message transcripts, or synthetic discovery logs, not from the business records of real deployments. To our knowledge, whether model-router decisions reveal sensitive topics on real user traffic, per request or across a user's requests, and what closing that channel costs, has not been measured.

This paper measures it. We treat the per-request routing decision, as recorded by a gateway, as a channel from the request's topic to
an observer who never sees the text (Fig.~\ref{fig:overview}). Four things make the measurement hard. Sensitive requests differ from others in length, and token counts are often visible as well, so
leakage must be measured at matched length and \emph{beyond} what length reveals. Labels for sensitive topics are noisy: the LLM judges we
use for medical content failed our own validity gate twice, against the author's labels on LMSYS and an AI coding assistant's on
WildChat. Public chat logs repeat the same prompts many times, which narrows
confidence intervals and inflates apparent replication: 61\% of the held-out harassment requests repeat a prompt from the exploration files.
And an honest measurement must survive the choice of router, dataset, and operating point, which invites selective reporting. We address
these with length-matched gaps and a flexible length model, moderation labels shipped with the data for the primary categories
(unvalidated, but needing no LLM judge), prompt fingerprints used as a duplicate proxy, and pre-registered protocols with a held-out replication of the primary per-request claims
(Section~\ref{sec:method}).

\medskip
\noindent Concretely, this paper contributes:
\begin{enumerate}
\item \textbf{The routing decision as a privacy channel.} A threat model grounded in field-level evidence from eleven gateways, routers,
and cloud platforms: in at least six of them, per-request records naming the chosen model can be kept without content and linked to
the caller, in several of them by default (Section~\ref{sec:threat}).
\item \textbf{A pre-registered, duplicate-aware measurement on real traffic.} At matched length, the direction of the shift depends on
the category and the router. For RouteLLM at the 50\% operating point, harassment and self-harm requests reach the strong model
19 points less often than comparable requests on never-seen prompts (post hoc; 24--32 points on all held-out rows, as pre-registered),
sexual requests 10 points more often (secondary; predicted to stay below 10), and medical requests, with exploratory LLM
labels, 31 points less often on distinct WildChat prompts (post hoc; 33 on all rows). The second cost/quality router,
notdiamond-0001, sends all four less often, though its never-seen self-harm gap is indistinguishable from zero. Twenty RouteLLM
decisions separate frequent medical askers from users with no medical requests at an area under the ROC curve (AUC) of 0.71,
exploratory and below the pre-registered primary endpoint's 0.75, and 0.92 when the decision is a domain router's health label (an
upper estimate, as its labels helped select what the judges saw; Sections~\ref{sec:perreq}--\ref{sec:userlevel}).
\item \textbf{Defenses and their cost.} Length-matched parity per category removes the per-request gap on real traffic when labels are
accurate (half of it with a cheap classifier). On RouterBench, the routers' gaps on sensitive subjects (13--42 points, pre-registered)
exceed those of an oracle that routes by realized accuracy gain (1--11, post hoc), and per-category parity costs at most 0.2 accuracy
points but leaves gaps of up to 0.064 (post hoc). Pooled parity fails because categories shift by different amounts and, for RouteLLM on LMSYS, in
opposite directions; per-conversation stickiness barely helps, and per-user budget bands leak at their edges. A post hoc exact per-user
rate hides only the strong count, and only at even request counts, at the cost of most of the router's self-assessed value, which
we could not express in accuracy and dollars. Its ordered log preserves odd-position decisions exactly
(Remark~\ref{rem:exact}); a post hoc attack using them reaches AUC 0.73 after 20 RouteLLM requests, exploratory
(Section~\ref{sec:oddattack}).
\end{enumerate}

Section~\ref{sec:related} reviews related work and Section~\ref{sec:method} the method; Sections~\ref{sec:discussion}--\ref{sec:limits}
discuss recommendations, ethics, and limitations.

\section{Background and Related Work}
\label{sec:related}

This section classifies prior work by what it protects or measures, and closes with a feature comparison (Table~\ref{tab:related}).

\subsection{Cost/quality and domain routing}
Routers choose, per request, which model answers, mostly to trade answer quality against cost. FrugalGPT learns which combination of models in a cascade to query for each request~\cite{chen2023frugalgpt}, and AutoMix escalates a query to a larger model when a smaller model's self-verified answer looks unreliable~\cite{aggarwal2023automix}. Hybrid LLM routes each query to a small or a large model by its predicted difficulty~\cite{ding2024hybrid}. RouteLLM trains routers between a strong and a weak model on human preference data from Chatbot Arena~\cite{ong2024routellm,chiang2024arena}, and RouterBench collects per-model outcomes and costs for evaluating routers~\cite{hu2024routerbench}. Commercial routers appeared alongside them: Not Diamond released notdiamond-0001, an open router between GPT-3.5 and GPT-4~\cite{notdiamond0001}, and OpenRouter's Auto Router assigns each request one of about 30 task types before choosing a model~\cite{openrouterauto}. %
Domain routers such as vLLM Semantic Router classify the request's domain and route by the label~\cite{vllmsr}. Cloud platforms offer routers of their own, such as intelligent prompt routing in Amazon Bedrock, which chooses between two models of one family by their predicted quality difference, and the model router of Microsoft Foundry, formerly Azure AI Foundry (their logging is described in Section~\ref{sec:threat}). %
An evaluation of four open-source routers on four benchmarks found that only vLLM Semantic Router's choice varies materially with the prompt, whereas RouteLLM's similarity-weighted ranking (\texttt{sw\_ranking}) router, at its default threshold of 0.5 and with a strong/weak pair it was not calibrated for, sent every task to the weak model~\cite{kumar2026routers}. %
Closest to our observation, Kassem et al.~\cite{kassem2025robust} evaluate five routers, including two commercial ones, on a benchmark of
task categories and find that routing decisions are driven by the query's category, for example sending adversarial prompts to weaker
models. In their privacy case study, Kassem et al. use a 200-prompt, four-category subset derived from the PUPA benchmark introduced by
PAPILLON~\cite{papillon2025}; it includes healthcare, and although the authors call its routing balanced, their Fig.~11 shows 146 of
the 200 prompts, 33 of the 50 healthcare ones, sent to the strong model, an unconditional share not comparable with a length-matched gap; their concern is
robustness, safety, and the routing of personal data, not what an observer of the decisions learns. Koul~\cite{koul2026register} shows that
complexity-based routing assigns non-standard English registers to lower-capacity tiers, largely through input length, and Wang et
al.~\cite{wang2026misguided} find on 800 benchmark queries that most of six evaluated commercial routers---including OpenRouter's
Auto Router and a vLLM Semantic Router fitted on the authors' training data---send nearly all queries from each source benchmark to
one model, largely independently of within-source difficulty. These works
measure what routing saves, what it costs in quality, and whom it serves worse; none asks what the decision reveals about the request
to those who record it.

\subsection{Privacy-preserving routing}
Systems that combine a local and a remote model for privacy decide, per request, what reaches the remote side: PAPILLON uses a local
model to protect what is sent to a proprietary one~\cite{papillon2025}, and PRISM routes between edge and cloud according to the
request's sensitivity~\cite{prism_router}. When such a decision depends on sensitivity, it is itself a sensitivity signal. Wu et
al.~\cite{pproute2026} protect the user query and routing model during route computation with secure computation, but do not study
exposure or logging of the selected model after computation.
AVEC~\cite{gaikwad2025avec} observes that the local/remote delegation bit leaks and bounds the leakage with randomized response, without
a measurement on real routers or traffic.

\subsection{Side channels of adaptive inference and LLM traffic}
Adaptive computation leaks through what it chooses to do. Kannan et al.~\cite{multiexit2023ccs} show that exit points of distributed
multi-exit networks reveal predictions and propose defenses that balance exits across classes, the structural precedent of our parity
defense. Akinsanya and Brennan~\cite{akinsanya2024ndss} show that timing channels of adaptive networks leak sensitive input
attributes. Whisper Leak~\cite{whisperleak2025} infers the topic of LLM conversations from the sizes and timing of streamed packets, and the
response timing of efficient inference, such as speculative sampling, reveals conversation topics too~\cite{carlini2024timing}.
In anonymity systems, path selection that depends on the client's location leaks it, and the leakage accumulates over repeated choices;
Tempest~\cite{wails2018tempest} measures this and CLAPS~\cite{rochet2020claps} trades it off against performance with linear programming.
Elsewhere, the ads a platform delivers to a user, aggregated over many impressions, reveal private attributes of real
users~\cite{chen2025adsprofiles}; expert routing inside mixture-of-experts LLMs leaks prompt attributes through hardware side
channels~\cite{ding2025moecho}, and the GPU telemetry it induces identifies harmful prompts~\cite{lv2026routescan}.
Pouryousef et al.~\cite{pouryousef2026traffic} show that a passive network observer can identify which model served a request from
encrypted traffic alone, so a routing decision may be visible even without logs. In agent systems, the same insight has recently been
applied to selections made by orchestrators rather than routers. Which service agents are called, and in what order, reveals intent and
can be masked with budgeted decoy calls~\cite{peng2026agentic}; communication-graph metadata~\cite{dangol2026graph} and, in a synthetic study, content-free
agent-discovery metadata~\cite{kim2026agentcardleak} reveal task types and intent; and choices among equivalent messages form selection channels that
length-matched controls do not close~\cite{xu2026sicc}. Outside agent systems, in a controlled support-routing task, an active attacker who writes the input can turn the
decision labels of Jev, a fast decision
model that applications use to route requests, into an oracle on hidden user attributes~\cite{wang2026jev}.
Our setting differs from all of these. The decision is made by a cost/quality or domain router from the request content, a passive observer reads it
from business records (logs and bills) that hold no content, rather than from network traffic, message transcripts, hardware
telemetry, or synthetic discovery logs, and we measure it on real user traffic, per request and per user, and
evaluate defenses against it. We therefore do not claim to be the first to treat a selection as a leak; to our knowledge, ours is the first
privacy measurement of model-router choices as observable log or billing metadata on real user traffic, including per-request and
per-user inference and defense costs.

\subsection{Defense lineage}
Our parity defense applies fairness post-processing to a router: group-dependent
thresholds equalize decision rates. Related optimal post-processing results exist
for equalized odds and equal opportunity~\cite{hardt2016equality} and, under a
Bayes-optimal-score assumption, for demographic parity~\cite{xian2023fair};
we do not claim that those results make our router rule optimal.
Our exact per-user-rate defense also keeps persistent per-user state; this only
loosely echoes RAPPOR, which memoizes a permanently randomized answer and then
applies fresh instantaneous randomization to each report~\cite{erlingsson2014rappor}.
It does not inherit RAPPOR's differential-privacy guarantee. We claim no new theory; our contribution is to instantiate these tools for the
router observer and to measure them on real traffic.

\begin{table*}[!t]
\centering
\caption{Related work by the properties this paper's contribution turns on. Real router: an implemented router or selection mechanism,
not a simulated one. Real traffic: evaluated on passively collected or released records of organic
user activity, rather than researcher-issued or replayed requests. A check means that at least one cited work in the row has the
property; partly means that it holds only in a weaker form (notes a--f).}
\label{tab:related}
\small
\begin{tabular}{@{}p{0.34\textwidth}cccccc@{}}
\toprule
Work & \shortstack{Real\\router} & \shortstack{Real\\traffic} & \shortstack{Decision\\as channel} & \shortstack{User-level\\accumulation} & \shortstack{Log/billing\\observer} & \shortstack{Defense\\evaluated} \\
\midrule
Cost/quality routers~\cite{chen2023frugalgpt,ding2024hybrid,ong2024routellm} & \checkmark & partly$^{\mathrm{a}}$ & -- & -- & -- & -- \\
Routing robustness and bias~\cite{kassem2025robust,koul2026register,wang2026misguided} & \checkmark & partly$^{\mathrm{b}}$ & partly$^{\mathrm{c}}$ & -- & -- & -- \\
Privacy-preserving routing~\cite{papillon2025,prism_router,pproute2026} & \checkmark & partly$^{\mathrm{b}}$ & -- & -- & -- & content only \\
AVEC~\cite{gaikwad2025avec} & -- & -- & \checkmark & partly$^{\mathrm{d}}$ & -- & simulation \\
Multi-exit networks~\cite{multiexit2023ccs} & \checkmark & -- & \checkmark & -- & -- & \checkmark \\
Tor path selection~\cite{wails2018tempest,rochet2020claps} & \checkmark & \checkmark & \checkmark & \checkmark & -- & \checkmark \\
Ad-delivery profiling~\cite{chen2025adsprofiles} & \checkmark & \checkmark & \checkmark & \checkmark & -- & -- \\
Agent selection leakage~\cite{dangol2026graph,xu2026sicc,peng2026agentic,kim2026agentcardleak} & \checkmark & -- & \checkmark & -- & partly$^{\mathrm{e}}$ & \checkmark \\
Decision-model attribute inference~\cite{wang2026jev}
& \checkmark & -- & \checkmark & -- & -- & -- \\
LLM traffic analysis~\cite{whisperleak2025,pouryousef2026traffic} & -- & -- & partly$^{\mathrm{f}}$ & -- & -- & \checkmark \\
This work & \checkmark & \checkmark & \checkmark & \checkmark$^\ast$ & \checkmark & \checkmark \\
\bottomrule
\end{tabular}

\smallskip
\footnotesize $^{\mathrm{a}}$RouteLLM trains on Chatbot Arena preference data but is evaluated on benchmarks; Hybrid LLM is evaluated on
MixInstruct, which includes a small share of shared ChatGPT conversations.
$^{\mathrm{b}}$Only on benchmarks built from public collections of user-LLM interactions: WildChat subsets, including PUPA, in Kassem et
al.; PUPA in PAPILLON. PPRoute~\cite{pproute2026} evaluates on EmbedLLM, MixInstruct (note a), and RouterBench.
$^{\mathrm{c}}$Decisions are shown to depend on query category or register, but not as a leak to an observer.
$^{\mathrm{d}}$Cumulative loss is tracked in a R\'enyi differential-privacy budget; no inference across a user's requests is measured.
$^{\mathrm{e}}$Loggers are among the modeled attackers (SICC~\cite{xu2026sicc}; Dangol~\cite{dangol2026graph}; AgentCardLeak~\cite{kim2026agentcardleak}), but the evaluations use
message transcripts, the transport metadata of agent traffic, or synthetic discovery logs, not business records of real deployments.
$^{\mathrm{f}}$The serving model is inferred from encrypted traffic, which behind a router is the decision, but no router is studied.
$^\ast$Medical requests on WildChat only, with exploratory LLM-judge labels; the moderation-labeled categories come from LMSYS-Chat-1M,
which has no user identifiers (Section~\ref{sec:userlevel}).
\end{table*}

\section{Threat Model}
\label{sec:threat}

This section states what a routed deployment records, who can read that record, and what the attacker is assumed to know. It rests
on a field-level survey of eleven gateways, routers, and cloud platforms, based on 59 primary sources (official documentation and
source-code schemas) accessed on October~1, 2026; the evidence table (in Traditional Chinese), with a link to the most recent Wayback
Machine snapshot that the availability API returned for 55 of the sources, is released with the artifact.

\subsection{System model}
Each request $x_t$ of a caller passes through a router that selects the model $R_t$ that answers it; for the two-model routers of
Section~\ref{sec:model}, $R_t$ is the binary strong-or-weak decision. A gateway or cloud platform can keep, per request, the record
(time, caller, $R_t$, input tokens, output tokens, cost), where the caller is an API key, organization member, user id, or AWS
Identity and Access Management (IAM) principal. In at least six of the eleven systems (OpenRouter, LiteLLM, Kong, Cloudflare, Helicone, and AWS), first-party documentation or
schemas show that per-request metadata with the served (for AWS, the requested) model and time can be kept without content and linked
to the caller. Portkey may be a seventh, but its documentation does not state whether no-content mode keeps the model field; for Azure we confirmed such records only in
its separate API Management gateway, which we do not count.
In OpenRouter and LiteLLM this metadata-only record is the default, and AWS CloudTrail records by default the caller and the requested
model of every Bedrock \code{InvokeModel} and \code{Converse} call, streaming variants included, without token counts; other inference
operations, such as asynchronous invocation, are data events, which are off by default.
Five routing products (OpenRouter, LiteLLM, Microsoft Foundry, Amazon Bedrock prompt routing, and Not Diamond) also return or record
the decision as a named field, such as Bedrock's \code{invokedModelId} or LiteLLM's complexity-router tier header, so reading $R_t$
requires no inference; OpenRouter's Auto Router also returns a task type when the caller requests metadata~\cite{openrouterauto}.

\subsection{Observers}
\begin{table}[!t]
\centering
\caption{Observers of routing metadata in the eleven surveyed systems. Default: available without any non-default configuration.}
\label{tab:observers}
\small
\setlength{\tabcolsep}{4pt}
\begin{tabular}{@{}>{\raggedright\arraybackslash}p{0.22\columnwidth}>{\raggedright\arraybackslash}p{0.37\columnwidth}%
>{\raggedright\arraybackslash}p{0.33\columnwidth}@{}}
\toprule
Observer & Sees & Default \\
\midrule
Org admin / platform team & per request: model, time, tokens, cost, key or member & OpenRouter, LiteLLM, Portkey: yes; Cloudflare,
Helicone: no caller id; Kong, AWS: after setup \\
Security team & per request: caller identity, model, time & CloudTrail history: yes; SIEM: after export \\
Finance & aggregates by model and hour or day, per key, member, or principal on some platforms & partly \\
Vendor; observability subprocessor & per-request metadata, also in privacy modes & hosted vendor: yes; subprocessor: once
configured \\
Network observer & no documented cleartext model id; serving model inferable from encrypted
traffic~\cite{pouryousef2026traffic} & not measured here \\
\bottomrule
\end{tabular}
\end{table}
The main adversary is an insider with default administrative access to the gateway (Table~\ref{tab:observers}), such as a LiteLLM
proxy admin or an OpenRouter organization admin: without enabling any content logging, they obtain each user's itemized sequence of
routing decisions and can run the itemized-log attacker A2 (\ref{eq:a2}). In OpenRouter organizations, every member, not only admins,
sees the model, cost, and timing of all members' usage. The security team sees the same sequence of
decisions once logs are exported to a security information and event management (SIEM) system, and on AWS, through CloudTrail's default event history and without any Bedrock
configuration, when the application itself routes between model identifiers with \code{InvokeModel} or \code{Converse}; forwarding
these events to a SIEM requires a trail. Finance sees aggregates by model and hour or day, per key, member, or principal
on some platforms, which corresponds to the bill-only attacker A1 (\ref{eq:a1}) and the billing observers of
Section~\ref{sec:attackers}. Hosted gateway vendors store the metadata themselves, and observability subprocessors still receive the
model, tokens, cost, and user or session identifiers in privacy modes that strip the content. A passive network observer can infer
the serving model from encrypted traffic~\cite{pouryousef2026traffic}; we do not measure this observer.

\subsection{Attacker capabilities and assumptions}
The attacker (i) sees $R_t$ per request or, in the case of finance, per-period aggregates; (ii) is conservatively given the token count
of each request's message, a stand-in for the input-token count that many surveyed per-request logs store alongside the model and
some finance records report in aggregate by model; (iii) never sees content; (iv) holds labels for some users, from which it fits
its likelihoods; and (v) needs no router internals, only the recorded outputs. It only reads records. Assumption (ii) fixes what we
measure: what routing reveals beyond the length of the request message, namely the gap at matched length (\ref{eq:dadj}) and the gain
over a length model (\ref{eq:dauc}), not the raw difference in strong-routing rates; logged input counts also cover earlier turns,
and we ignore output tokens. Assumption (iv) is generous, as our attackers are fitted on half of
the users (Section~\ref{sec:attackers}), but an insider could approximate such labels from their own or colleagues' known usage.

\subsection{Why not enable content logging}
An insider able to switch on content logging might seem to need no side channel, yet three situations make the routing record the
channel that matters. First, where policy or compliance keeps content logging off, enabling it breaks that policy and captures only
future requests, whereas the routing record of past requests is already kept. Second, metadata is kept independently of content and
often longer: Helicone's self-hosted schema sets a three-month expiry only on request and response bodies, LiteLLM by default never
purges its spend logs, and CloudTrail's event history keeps the past 90 days of management events per region without any setup.
Third, some observers cannot switch content logging on at all: finance, observability subprocessors in privacy modes, and anyone who
breaches or subpoenas a metadata store see only the metadata. Two harms follow. An employer that administers the gateway could try to
infer which employees often ask about health or self-harm, while the employees may believe that disabled content logging protects
them; and retained metadata could be breached, misused, or subpoenaed after the content is deleted. Per-request leakage
(Section~\ref{sec:perreq}) lets an observer flag individual requests; user-level accumulation, which we measure for medical requests
only (Section~\ref{sec:userlevel}), lets it rank users.

\subsection{Framing}
Many gateways, and some cloud platforms, record such metadata by design for billing and monitoring; their content controls are
documented to remove content, not metadata; we do not treat any of this as a vulnerability. Our point is narrower: when the router
chooses the model from the content, disabling content logging does not by itself hide the topic. Latency and timing channels are out
of scope.

\section{Model and Methodology}
\label{sec:method}

The original measurements are taken or recomputed from the released result files, without altering them.
The post hoc odd-position follow-up is recorded separately and reported in full in Table~\ref{tab:oddattack}.

\subsection{System and observer model}
\label{sec:model}
Requests are indexed by $i=1,\dots,N$. Request $i$ has a caller $u_i$, a time $t_i$, content $x_i$, a token count $n_i$ (the length
of its message under the RouteLLM tokenizer), and a category $c_i \in \mathcal{C}\cup\{O,U\}$, where $\mathcal{C}$ is the set of
sensitive categories, $O$ the comparison set of non-sensitive requests, and $U$ the requests in neither (e.g., law-domain requests on
WildChat), which enter no gap. On the LMSYS held-out files, which no judge read, $O$ also contains medical and legal requests (on the
exploration files, adding them to $O$ moved RouteLLM's harassment, self-harm, and sexual gaps by less than 0.01). A router is a score function $s$ that maps content to $[0,1]$, followed by a threshold; we write
$s_i = s(x_i)$. At operating point $q$ the threshold $\tau_q$ is the $(1-q)$-quantile of $s$ over $O$, so that a fraction $q$ of
comparison requests reaches the strong model, and the decision is $R_i = \mathbf{1}[s_i\ge\tau_q]$. Unless stated otherwise, $q=0.5$,
the example in RouteLLM's documented calibration to a target fraction of strong-model calls~\cite{routellmrepo}; Section~\ref{sec:perreq}
reports the whole range of $q$. The gateway keeps the record
$(u_i, t_i, R_i, n_i, \mathrm{cost}_i)$; $n_i$ stands in for the logged input tokens, which can also include the conversation
history, and we ignore output tokens. The observer reads records and never reads $x_i$ (Section~\ref{sec:threat}).

For RouteLLM, $s$ is the probability that the strong model's answer wins (one minus the softmax mass of the two non-winning
classes); for notdiamond-0001, that of the GPT-4 class under the router's own prompt template. The domain router's decision is whether
its domain label is \emph{health}.

\subsection{Data, routers, and labels}
\label{sec:data}\label{sec:labels}
From English conversations we use the user messages of at least 10 characters: all of them in WildChat-1M~\cite{zhao2024wildchat}
(982,067 requests from 97,742 hashed-IP users) and the first per conversation in LMSYS-Chat-1M~\cite{zheng2024lmsys} (251,230 from
files 0--1 for exploration and 502,477 from the unused files 2--5 for the held-out replication). We also use
RouterBench~\cite{hu2024routerbench} (36,497 benchmark prompts with per-model correctness and cost). The routers are RouteLLM's BERT router~\cite{ong2024routellm},
notdiamond-0001~\cite{notdiamond0001} (\emph{notdiamond} below), and the Vela domain classifier of vLLM Semantic
Router~\cite{vllmsr} (the \emph{domain router}), all at pinned revisions whose weight hashes were checked before every run. Text was processed on a GPU host
with open-weight models, and only text-free tables left it (exceptions: Section~\ref{sec:ethics}). We name the studies we report after their protocols (E-3 to E-9,
Table~\ref{tab:prereg}); E-1 and E-2 were exploratory pilots whose results we do not report.

\begin{table*}[!t]
\centering
\caption{Pre-registration record: every primary and co-primary endpoint of E-3 to E-9, including failed ones, and the label gates.
Frozen: date (UTC+8) of the protocol-freezing commit in our private repository, before data download or collection; the WildChat studies
re-analyzed one release, each frozen after its predecessor's results (E-5: before judging; E-8: before parity and band defenses).}
\label{tab:prereg}
\footnotesize
\setlength{\tabcolsep}{4pt}
\begin{tabular}{@{}llp{0.335\textwidth}p{0.49\textwidth}@{}}
\toprule
Study & Frozen & Pre-registered endpoint & Outcome \\
\midrule
E-3 & 2026-09-28 & Medical $|\Delta_{\mathrm{adj}}| \ge 0.10$ (P1) & Met: $-0.251$ \\
E-3 & 2026-09-28 & User-level AUC $\ge$ 0.75 at $T{=}20$ (P2) & Not determinable (15 positive users; 30 required) \\
E-4 & 2026-09-29 & Medical $|\Delta_{\mathrm{adj}}| \ge 0.10$ (P1) & Met: $-0.204$ (zero-shot labels; precision 0.55 in an AI assistant's check) \\
E-4 & 2026-09-29 & User-level AUC $\ge$ 0.75 at $T{=}20$ (P2) & Met by rule: 0.7515 [0.699, 0.799]; later shown to be the upper end of noise \\
E-4 & 2026-09-29 & Same AUC, labels from later requests (co-primary) & Not met: 0.743 [0.673, 0.809] \\
E-5 & 2026-09-29 & Medical $|\Delta_{\mathrm{adj}}| \ge 0.10$ (P1) & Met: $-0.331$ \\
E-5 & 2026-09-29 & User-level AUC $\ge$ 0.75 at $T{=}20$ (P2) & Not met: 0.707 [0.671, 0.740] \\
E-5 & 2026-09-29 & Same AUC, labels from later requests (co-primary) & Not met: 0.688 [0.637, 0.733] \\
E-5 & 2026-09-29 & Label gate: precision $\ge$ 0.85, recall $\ge$ 0.80 & Not met: precision 0.70 against labels assigned by the same AI assistant (Claude), not by the author; recall not estimable; medical results exploratory \\
E-6a & 2026-09-30 & Sexual and self-harm $|\Delta_{\mathrm{adj}}| \ge 0.10$, CI excludes 0 & Partly met: self-harm $-0.171$; sexual $+0.088$ (below 0.10; predicted $<0$) \\
E-6a & 2026-09-30 & Label gate: precision $\ge$ 0.85, weighted recall $\ge$ 0.80 & Not met: precision 0.75, weighted recall 0.69 (against the author's labels) \\
E-7-lite & 2026-10-01 & Harassment, self-harm $\Delta_{\mathrm{adj}} \le -0.10$, 97.5\% CI $<0$ & Both met: $-0.316$, $-0.238$; never-seen prompts (post hoc): $-0.193$, $-0.186$, still met \\
E-8 & 2026-09-30 & Deployable defense (four criteria) & Not met: leak, user AUC, retained benefit (0.33), and budget all fail \\
E-8b & 2026-10-01 & Pooled parity: $\le$1\,pp acc., $\le$5\% cost, leak $\le$ 0.02 & Not met: accuracy $-0.22$\,pp and cost $-3.8$\% pass; leak 0.137 fails \\
E-9 & 2026-10-01 & Commercial router, F2 topic templates, default cost tier: model-choice TV, HEALTH vs.\ NEUTRAL & Pending; established if net TV $\ge$ 0.10 and Holm-adjusted permutation $p<0.05$ for at least one of the two comparisons \\
E-9 & 2026-10-01 & Same, MENTAL vs.\ NEUTRAL & Pending. TV: total variation distance; net TV: TV minus its permutation mean \\
\bottomrule
\end{tabular}
\end{table*}

Harassment, self-harm, sexual, and violence groups use the per-message OpenAI
moderation outputs shipped with LMSYS-Chat-1M~\cite{zheng2024lmsys,lmsyscard}; conversations flagged for sexual
content involving minors are dropped entirely. Medical labels come from LLM judges: one judge on WildChat and the intersection of two on LMSYS. Judges saw only
\emph{candidates} (a zero-shot natural-language-inference medical score of at least 0.3 or a Vela label of health, biology, or psychology),
on WildChat also the other messages of their conversations, and a 2\% random sample of the rest; messages no judge saw count as non-medical.
Against stratified gold sets, the judges' precision was 0.75 [0.63, 0.84] on LMSYS, where the author labeled 150 items, and 0.70 on
WildChat, where the AI coding assistant used in this study (Claude), not the author, assigned the 400 reference labels, not blind,
from at most 300 characters of each message and 120 of the previous one, with terminal lines capped at 330 (judge: up to 1,000
and 300; Section~\ref{sec:ethics}). Both fall short of the pre-registered 0.85, so medical results are exploratory. The assistant also
judged the spot checks of zero-shot labels (score $\ge 0.8$) of E-3 and E-4, from at most 200 characters each, which gave E-4's
precision of 0.55 (Table~\ref{tab:prereg}).

\subsection{Per-request leakage}
\label{sec:perreqdef}
For a category $c$, let $b^{(c)}_i\in\{1,\dots,10\}$ be the decile of $\log n_i$ with edges taken from the requests of $c$. Write
$p_c = P(R{=}1\mid c)$, $p_O = P(R{=}1\mid O)$ ($=q$ by construction for the two score routers), $w_{c,k} = P(b^{(c)}{=}k\mid c)$, $p_{O,k} = P(R{=}1\mid O,\, b^{(c)}{=}k)$, and let $K_c$ be the bins that
contain comparison requests. The \emph{length-matched gap} is
\begin{equation}
\dadj(c) = p_c - \tilde p_O(c), \qquad \tilde p_O(c) = \frac{\sum_{k\in K_c} w_{c,k}\, p_{O,k}}{\sum_{k\in K_c} w_{c,k}},
\label{eq:dadj}
\end{equation}
the difference between the group's strong rate and that of comparison requests with the group's length profile. The
\emph{per-request privacy loss at matched length} is
\begin{equation}
\varepsilon(c) = \max\left\{\left|\log\frac{p_c}{\tilde p_O(c)}\right|,\; \left|\log\frac{1-p_c}{1-\tilde p_O(c)}\right|\right\},
\label{eq:eps}
\end{equation}
the maximum absolute log-likelihood ratio between the two Bernoulli distributions after averaging over length. Opposing effects
across bins can cancel, so zero gap or zero $\varepsilon$ does not establish conditional independence given length. A maximum over
bins would give a worst-bin distributional bound, not differential privacy for individual inputs. Earlier result files, including those
behind Fig.~\ref{fig:perreq}b, store an $\varepsilon$ against the unmatched comparison rate (fields \code{eps\_per\_request},
\code{eps\_health}, \code{eps}); we recompute matched-length values from $p_c$ and $\dadj$. The
\emph{gain over token counts} is
\begin{equation}
\Delta\mathrm{AUC}(c) = \mathrm{AUC}(\ell, R) - \mathrm{AUC}(\ell),
\label{eq:dauc}
\end{equation}
where $\ell$ is the one-hot encoding of 20 quantile bins of $\log n$ and each AUC is that of a logistic regression (Newton's method,
$L_2$ penalty $10^{-4}$) under 5-fold cross-validation, on the requests of $c$ and a 20:1 random subsample of $O$.
The joint model adds one routing coefficient, without routing-by-length interactions; its gain does not exhaust what a
length-aware observer could infer.

\begin{remark}[single decision]
\label{rem:single}
An observer who sees only $R$ distinguishes $c$ from $O$ with $\mathrm{AUC}_R = \tfrac12 + \tfrac12|p_c - p_O|$.
\end{remark}
\begin{IEEEproof}
Without loss of generality $p_c\ge p_O$ (otherwise score by $1-R$). Then
$\mathrm{AUC} = P(R_c > R_O) + \tfrac12 P(R_c = R_O) = p_c(1-p_O) + \tfrac12[p_c p_O + (1-p_c)(1-p_O)] = \tfrac12 + \tfrac12(p_c-p_O)$.
\end{IEEEproof}
\noindent For RouteLLM on WildChat medical requests, $|p_c-p_O| = 0.308$ gives 0.654; the cross-validated AUC is 0.651.

\begin{remark}[ceiling]
\label{rem:ceiling}
$\Delta\mathrm{AUC}$ depends on what length already reveals. In a latent binormal-score approximation, given the class, length
($\mathrm{AUC}(\ell)\ge\tfrac12$) and the decision (AUC $\tfrac12+\tfrac12|\dadj|$ of Remark~\ref{rem:single} at every length) are
independent equal-variance normal scores, whose optimal combination gives
\begin{equation}
\Delta\mathrm{AUC} \approx \Phi\!\left(\sqrt{z_\ell^2 + z_\Delta^2}\right) - \Phi(z_\ell),
\label{eq:ceiling}
\end{equation}
with $z_\ell = \Phi^{-1}(\mathrm{AUC}(\ell))$ and $z_\Delta = \Phi^{-1}(\tfrac12 + \tfrac12|\dadj|)$,
which is strictly decreasing in $\mathrm{AUC}(\ell)$ for a fixed nonzero gap.
\end{remark}
\noindent In a post hoc check on the nine category-by-split cases on LMSYS, (\ref{eq:ceiling}) predicts the observed $\Delta\mathrm{AUC}$ to within 0.013.
Under this approximation, a small $\Delta\mathrm{AUC}$ can coexist with a nonzero routing gap when length is informative; it does
not by itself identify the reason for a small measured increment. We report $\dadj$ and $\varepsilon$ as matched-length summaries
and $\Delta\mathrm{AUC}$ as the specified attacker's gain over token counts, not a distribution-free information bound.

\subsection{Estimation and replication}
\label{sec:estimation}
On WildChat we resample users with replacement (1,000 draws) and recompute (\ref{eq:dadj}) from per-user counts by length bin, so
that requests of the same user are not treated as independent. LMSYS has no user identifiers, so the pre-registered intervals resample
requests. Public chat logs repeat prompts, so we also resample fingerprint clusters. The fingerprint of request $i$ is
$\phi_i = (s^{\mathrm{RL}}_i, s^{\mathrm{ND}}_i, n_i)$, with both scores rounded to $10^{-6}$, a proxy for text identity rather than
a text-identity test. Throughout, \emph{distinct} and \emph{never-seen} refer to this proxy: the latter keeps held-out category
requests whose fingerprint occurs nowhere in exploration. Its intervals resample fingerprint clusters, with the comparison set (about
480,000 requests) held fixed. The two co-primary endpoints of the held-out replication use Bonferroni-adjusted 97.5\% intervals.

\medskip\needspace{4\baselineskip}\noindent\textbf{Procedure 1 (per-request leakage).}
\begin{enumerate}
\item Set $\tau_q$ on $O$ and compute every $R_i$; bin $c$ and $O$ by the deciles of $\log n$ over $c$; compute $\dadj(c)$ by
(\ref{eq:dadj}) and $\varepsilon(c)$ by (\ref{eq:eps}).
\item Resample as above (on RouterBench, the group's and the comparison set's prompts separately; 1,000 or 2,000 draws), recompute,
and take percentile intervals.
\item Compute $\Delta\mathrm{AUC}(c)$ by (\ref{eq:dauc}) on $c$ and the 20:1 subsample of $O$.
\end{enumerate}

The original studies in Table~\ref{tab:prereg} were pre-registered: their protocols, endpoints, thresholds, and predictions were committed to version control
before its data were downloaded or collected or, if it reused data, before its new step (E-5: judge labels; E-8: parity and budget
defenses, its two stickiness variants having been measured post hoc in E-5b). The WildChat
studies reused one dataset, each frozen after its predecessor's results were known. Deviations are logged; Table~\ref{tab:prereg} lists
every primary and co-primary endpoint with its freeze date. The odd-position follow-up is post hoc and exploratory, not a new
pre-registered endpoint; its comparisons were fixed before execution, after the exact-rate results and transcript identity were known.

\subsection{User-level attackers}
\label{sec:attackers}
Users with at least 10 requests are split into halves by a hash of their identifier. A user is positive if at least 20\% of their
requests belong to $c$, negative if none do, and is otherwise excluded. Training positives and negatives define the strong rates $a$ and
$b$ over all of their requests. The attacker sees the first $T$ records of a test user $u$,
the $t$-th in time order with decision $R_{u,t}$ and token count $n_{u,t}$. The \emph{bill-only} attacker uses the strong
count $k_u = \sum_{t\le T} R_{u,t}$:
\begin{equation}
\mathrm{A1}(u) = k_u \log\frac{a}{b} + (T-k_u)\log\frac{1-a}{1-b}.
\label{eq:a1}
\end{equation}
A1+L adds the mean log length in a logistic regression fitted on the training users. The \emph{itemized-log} attacker uses each record:
\begin{multline}
\mathrm{A2}(u) = \sum_{t\le T} \biggl[R_{u,t}\log\frac{\hat p_{1,\beta_t}}{\hat p_{0,\beta_t}}\\
+ (1-R_{u,t})\log\frac{1-\hat p_{1,\beta_t}}{1-\hat p_{0,\beta_t}}\biggr],
\label{eq:a2}
\end{multline}
where $\beta_t$ is the global decile of $\log n_{u,t}$ and $\hat p_{y,k} = (\sum R + \tfrac12)/(\sum 1 + 1)$ is computed over all
requests in decile $k$ of training users with label $y$ (1 positive, 0 negative). The \emph{topic-only null} predicts A1's AUC if each
request's strong probability depended only on its class ($c$, $O$, or $U$), estimated on training users and simulated 200 times.
Between-user heterogeneity is measured by the intra-class correlation $\rho$ of the routing decisions within each user label, over all of
the users' requests (ANOVA estimator for binary outcomes). \emph{Billing} observers see only the strong count, the total tokens, or the total
cost $\sum_{t\le T}\mathrm{cost}_{u,t}$, with a stylized $\mathrm{cost}_i = n_i\lambda^{R_i}$ for a strong-to-weak price ratio
$\lambda\in\{20,50\}$.

\medskip\noindent\textbf{Post hoc odd-position attacker (A2-odd).}
For the exact-rate rule and cheap-classifier parity followed by that rule, A2-odd uses (\ref{eq:a2}) restricted to odd positions:
both probability estimation on training users and scoring within a test user's first $T$ requests use only $t=1,3,\ldots$.
Thus $T\in\{5,10,20,50\}$ denotes total requests, not retained bits; the attacker uses $3,5,10,25$ decisions, respectively.
We retain the E-8 defense-study split, full-record labels, global length bins, original router thresholds, and probability smoothing.
Each router--defense pair has one training fit shared across horizons, using 409 positive and 6,520 negative training users.
The observer knows the complete ordered prefix and its starting phase. With parity first, the preserved input is parity's output,
not the raw router decision. We reproduce original A2 on the same test users before comparing it with A2-odd; there is no test-set
sign reversal or model selection. The AUC difference uses 1,000 paired, label-stratified bootstrap resamples of test users
(seed 20260928), holding training fits fixed. The 95\% percentile intervals are unadjusted, descriptive, and conditional on those fits.

\begin{remark}[saturation]
\label{rem:saturation}
If, given class $y$, a user's decisions are independent Bernoulli draws with a user-specific rate $\pi_u\sim F_y$ of mean $0<\mu_y<1$
and $\rho_y=\mathrm{Var}(\pi_u)/[\mu_y(1-\mu_y)]$, then $k_u/T\to\pi_u$ as $T\to\infty$. For $a\ne b$, (\ref{eq:a1}) is strictly monotone
in $k_u$ and A1's AUC tends to $P(\sigma\pi_1>\sigma\pi_0)+\tfrac12P(\pi_1=\pi_0)$, $\sigma=\mathrm{sign}(a-b)$, $\pi_y\sim F_y$ independent.
If $\rho_0=\rho_1=0$, the $F_y$ are point masses and the limit is one if $\sigma(\mu_1-\mu_0)>0$; with heterogeneity, below one if
$P(\sigma\pi_1\le\sigma\pi_0)>0$.
\end{remark}

\subsection{Defenses}
\label{sec:defdef}
A defense post-processes the router's scores into decisions $R'$, using, where noted, a category label and per-user state. Let
$B_1,\dots,B_{10}$ be the global deciles of $\log n$ and $r_k$ the mean of $R$ over $B_k$.
\begin{itemize}
\item \emph{Mixing} (D1, rate $\eta$): $R'_i = R_i$ with probability $1-\eta$; otherwise $R'_i$ is a fair coin.
\item \emph{Length-matched parity} (D2) over a partition $\mathcal{P}$ of the requests: for each bin $k$ and part $P\in\mathcal{P}$, let
$m_{P,k} = \mathrm{round}(r_k\,|P\cap B_k|)$, and set $R'_i=1$ for the $m_{P,k}$ highest-scoring requests of $P\cap B_k$.
\emph{Per-category} parity uses $\mathcal{P} = \mathcal{C}\cup\{O\cup U\}$, \emph{pooled} parity $\mathcal{P} = \{\bigcup\mathcal{C}\cup U,\, O\}$ ($U$ is empty on RouterBench). On LMSYS, whose moderation flags
can overlap, per-category parity is applied to one category $c$ at a time, with parts $c$ and all other requests. The category label comes from a classifier, either
the evaluation labels (perfect labels) or a cheap one.
\item \emph{Budget band} (D3, slack $\delta$, target $Z=\tfrac12$): process each user's requests in order. Let $k$ be the strong count
before request $t$ and $R$ the router's decision. Set $R'=0$ if $R=1$ and $k+1-Zt>\delta$; set $R'=1$ if $R=0$ and $k-Zt<-\delta$;
otherwise $R'=R$. The \emph{exact-rate} variant uses $\delta=\tfrac12$.
\item \emph{Stickiness}: the first request's decision is reused for the whole conversation (D4) or for all of the user's requests (D5).
\end{itemize}

\begin{proposition}[per-category parity]
\label{prop:parity}
Per-category parity (i) preserves each bin's budget up to rounding, (ii) gives every part the strong rate $r_k$ within every bin up to
rounding, so that on these bins any category's gap against $O\cup U$ vanishes up to rounding, and (iii) routes an upward-closed
top-score set in each nonempty cell $P\cap B_k$ (cutoff ties: later rows first).
\end{proposition}
\begin{IEEEproof}
(i) The parts partition $B_k$, so $\sum_P m_{P,k} = \sum_P \mathrm{round}(r_k|P\cap B_k|)$ differs from $r_k|B_k|$ only by
rounding. (ii) and (iii) hold by construction.
\end{IEEEproof}
\noindent The $\dadj$ reported after parity (Tables~\ref{tab:defrb} and~\ref{tab:defwild}) can keep a residual because it uses the
category's own deciles (up to 0.064 on RouterBench, where $U$ is empty) and, on WildChat, compares with $O$ rather than $O\cup U$.

\begin{remark}[pooling]
\label{rem:pooling}
Pooled parity constrains only the weighted average of category selection rates within each bin. At the pooled cutoff, individual
categories may remain above or below the comparison rate even when the pooled gap is zero. Their post-parity directions are not
determined by their undefended gaps alone, because category score distributions can cross.
\end{remark}

\begin{remark}[boundary leak]
\label{rem:band}
Under an iid Bernoulli input model with rate $\pi$, the reflected count process favors the lower boundary when $\pi<Z$ and the upper
boundary when $\pi>Z$, accounting for its even--odd phases. This explains a possible leak through boundary occupancy, not a guarantee
for arbitrary correlated request sequences. Section~\ref{sec:band} reports the empirical boundary distributions.
\end{remark}

\begin{remark}[exact-rate transcript]
\label{rem:exact}
For the exact-rate rule ($Z=\delta=\tfrac12$), initialized at $k_0=0$, let $R_t$ be its binary input decision and $R'_t$ its output.
For each pair, $R'_{2j-1}=R_{2j-1}$ and $R'_{2j}=1-R_{2j-1}$. Thus $k_{2m}=m$, but for a fixed complete ordered prefix of length
$2m$ and any user label $Y$,
\[
I(Y;R'_{1:2m})=I(Y;R_1,R_3,\ldots,R_{2m-1}),
\]
where $I$ is mutual information. At odd $T$, $k_T-(T-1)/2=R_T$.
\end{remark}
\begin{IEEEproof}
Inductively, $k_{2j-2}=j-1$ makes both choices feasible at request $2j-1$, whereas only its complement is feasible at $2j$.
Hence $k_{2j}=j$. Odd-coordinate projection and deterministic complementation are inverse maps between the two transcripts,
giving the information identity without an independence assumption.
\end{IEEEproof}
\noindent The count at even $T$ is uninformative, but the ordered log preserves the odd-indexed input decisions verbatim.
A2 does not model position; its measured performance is not an upper bound on residual inference.
Section~\ref{sec:oddattack} tests the position-aware observer above.

\medskip\needspace{4\baselineskip}\noindent\textbf{Procedure 2 (per-category parity).}
\begin{enumerate}
\item Compute the undefended decisions $R$ and the global length deciles $B_k$; set $r_k$ to the strong rate in $B_k$.
\item Form one part per sensitive category and one of all other requests ($O\cup U$); in every part $P$ and bin $k$, send the
$\mathrm{round}(r_k\,|P\cap B_k|)$ highest-scoring requests to the strong model and the rest to the weak model.
\end{enumerate}
\noindent An online version also needs calibrated score cutoffs for each category--length cell, not just target rates; it does not
inherit the batch rule's finite-sample parity under distribution shift. Our LMSYS evaluations treat categories separately and do not
establish simultaneous protection for overlapping categories. No per-user state or router retraining is needed for this post-processing.

\medskip
The router's own view of utility is $V(R) = \sum_i R_i(2s_i-1)$, its expected count of strong-routed requests where the strong
model is the better choice minus those where it is not (for RouteLLM, ties count as not). We report the
\emph{skill kept}
\begin{equation}
\mathrm{SK}(R') = \frac{V(R') - \bar R' N \bar g}{V(R) - \bar R N \bar g}, \qquad \bar g = \frac{1}{N}\sum_i (2s_i-1),
\label{eq:skill}
\end{equation}
with budgets $\bar R = N^{-1}\sum_i R_i$ and $\bar R'$. Its denominator, $\sum_i R_i(2s_i-1-\bar g)$, is positive for $0<\bar R<1$, as
$R$ routes the top scores. $\mathrm{SK}(R)=1$, and same-budget random routing has SK 0 only in expectation. It is a
self-assessment; Section~\ref{sec:realunits} checks it against real accuracy for per-request defenses.

\subsection{Real-unit evaluation}
\label{sec:realdef}
On RouterBench~\cite{hu2024routerbench}, $a^S_i, a^W_i\in[0,1]$ are the strong and weak models' scores on prompt $i$, and
$\mathrm{cost}^S_i, \mathrm{cost}^W_i$ their dollar costs. For decisions $R$ with budget $\bar R$,
\begin{align}
\mathrm{Acc}(R) &= \frac1N\sum_i \left[R_i a^S_i + (1-R_i)a^W_i\right], \label{eq:acc}\\
\mathrm{Cost}(R) &= \frac{1000}{N}\sum_i \left[R_i \mathrm{cost}^S_i + (1-R_i)\mathrm{cost}^W_i\right], \label{eq:cost}
\end{align}
and random routing at the same budget attains $\bar R\,\bar a^S + (1-\bar R)\,\bar a^W$. The post hoc \emph{oracle} routes by
realized accuracy gain: it sends the $\bar R N$ prompts with the largest $a^S_i-a^W_i$ to the strong model, breaking the many ties at
random. It is one quality-optimal router, not the only one: with per-category parity it keeps its accuracy, and its largest gap falls
to 0.031 or less (Table~\ref{tab:defrb}). We therefore compare gaps in points, not as ratios.

\section{Per-Request Leakage}
\label{sec:perreq}

This section reports how differently the router treats each sensitive category from comparable requests on real traffic
(Table~\ref{tab:perreq}, Fig.~\ref{fig:perreq}), whether the differences survive repeated prompts, and how they compare with an accuracy
oracle (Fig.~\ref{fig:excess}).

\begin{table*}[!t]
\centering
\caption{Per-request leakage at the 50\% operating point (router sends 50\% of non-sensitive requests to the strong model).
$\Delta_{\mathrm{adj}}$: strong-routing rate of the group minus the length-reweighted rate of the comparison set, with 95\% CI
(request bootstrap on LMSYS, user-cluster bootstrap on WildChat). $\varepsilon$: per-request privacy loss at matched length
(\ref{eq:eps}). $\Delta$AUC: observer-AUC gain of the routing decision over a flexible 20-bin length model.}
\label{tab:perreq}
\small
\setlength{\tabcolsep}{4pt}
\begin{tabular}{@{}lllrrlrrr@{}}
\toprule
 & & & & \multicolumn{4}{c}{RouteLLM} & notdiamond \\
\cmidrule(lr){5-8}\cmidrule(l){9-9}
Category & Data & Status & $n$ & $P(\mathrm{strong})$ & $\Delta_{\mathrm{adj}}$ [95\% CI] & $\varepsilon$ & $\Delta$AUC & $\Delta_{\mathrm{adj}}$ \\
\midrule
Harassment & LMSYS held-out, all & confirmatory & 7,140 & 0.187 & $-0.316$ [$-0.325$, $-0.307$]$^\dagger$ & 0.99 & $+0.014$ & $-0.276$ \\
Harassment & LMSYS held-out, new$^\S$ & post hoc & 1,786 & 0.351 & $-0.193$ [$-0.216$, $-0.169$]$^\dagger$ & 0.44 & $+0.010$ & $-0.166$ \\
Harassment & LMSYS exploration & exploratory & 3,563 & 0.179 & $-0.327$ [$-0.338$, $-0.316$] & 1.04 & $+0.013$$^\P$ & $-0.289$ \\
\addlinespace
Self-harm & LMSYS held-out, all & confirmatory & 276 & 0.297 & $-0.238$ [$-0.289$, $-0.190$]$^\dagger$ & 0.59 & $+0.061$ & $-0.097$ \\
Self-harm & LMSYS held-out, new$^\S$ & post hoc & 164 & 0.342 & $-0.186$ [$-0.258$, $-0.116$]$^\dagger$ & 0.43 & $+0.050$ & $-0.028$ \\
Self-harm & LMSYS exploration & pre-registered & 158 & 0.361 & $-0.171$ [$-0.243$, $-0.099$] & 0.39 & $+0.033$$^\P$ & $-0.133$ \\
\addlinespace
Medical & WildChat & pre-registered$^\ddagger$ & 32,215 & 0.192 & $-0.331$ [$-0.353$, $-0.308$] & 1.00 & $+0.107$$^\P$ & $-0.154$ \\
Medical & LMSYS exploration & exploratory$^\ddagger$ & 5,471 & 0.423 & $-0.162$ [$-0.172$, $-0.152$] & 0.33 & $+0.033$$^\P$ & $-0.219$ \\
\addlinespace
Sexual & LMSYS held-out, all & secondary & 7,713 & 0.803 & $+0.102$ [$+0.094$, $+0.110$] & 0.42 & $+0.006$ & $-0.054$ \\
Violence & LMSYS held-out, all & secondary & 1,759 & 0.735 & $+0.071$ [$+0.054$, $+0.088$] & 0.24 & $+0.004$ & $-0.054$ \\
\bottomrule
\end{tabular}

\smallskip
\footnotesize $^\dagger$97.5\% CI (Bonferroni over the two co-primary endpoints).
$^\S$New: never-seen rows (fingerprint absent from the exploration files), 2,750 harassment and 190 self-harm
rows; $n$: their distinct prompts; CI: bootstrap over distinct prompts;
$\Delta$AUC (not on this subset): all distinct held-out prompts (3,284 and 212).
$^\ddagger$LLM-judge labels that failed their pre-registered gate; in three label-error scenarios, simulated only for RouteLLM on
LMSYS, the sign held and the magnitude varied (Section~\ref{sec:limits}). $^\P$Post hoc, 20-bin length model; frozen log-length values:
$+0.178$, $+0.088$, $+0.169$, and $+0.042$ in row order.
\end{table*}

\begin{figure*}[!t]
\centering
\includegraphics[width=0.96\textwidth]{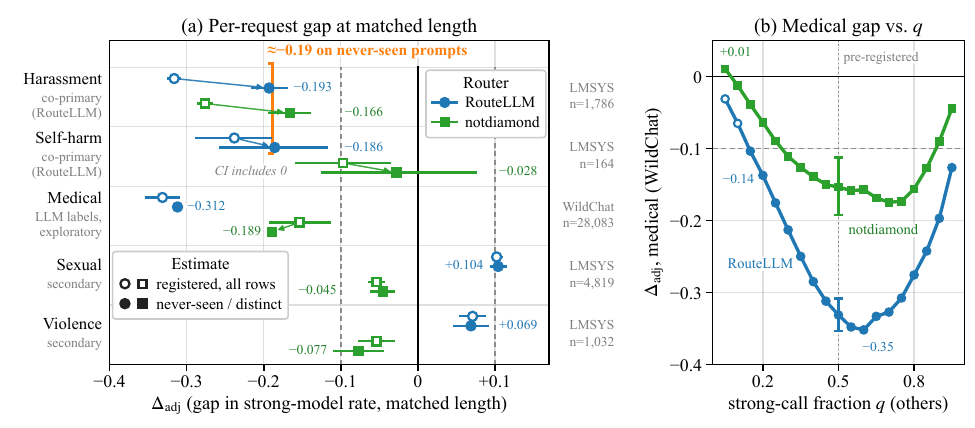}
\caption{The direction of per-request leakage depends on the category and the router, including on prompts never seen during exploration, and RouteLLM's WildChat medical gap is negative at all 19 operating points.
(a)~$\Delta_{\mathrm{adj}}$ (strong-routing rate of the category minus the length-reweighted rate of the comparison set) at the 50\% operating point; blue circles: RouteLLM, green squares: notdiamond.
Hollow markers are the pre-registered all-rows estimates, with intervals of the same type as in Table~\ref{tab:perreq} (which lists them for RouteLLM only). Harassment and self-harm: LMSYS held-out files; confirmatory co-primary endpoints for RouteLLM, secondary for notdiamond; 97.5\% request-bootstrap CI for both routers. Sexual and violence: LMSYS held-out, secondary, 95\% request-bootstrap CI. Medical: WildChat, 95\% user-cluster bootstrap CI; pre-registered but exploratory, because the LLM-judge labels failed their pre-registered validity gate.
Filled markers are post hoc. On LMSYS they use only prompts whose fingerprint never occurs in the exploration files, with a bootstrap over distinct prompts and the comparison set held fixed (97.5\% for harassment and self-harm, 95\% otherwise). On WildChat medical they use one row per distinct prompt, with a 95\% bootstrap over distinct prompts; that interval ignores user clustering and holds the comparison set fixed, so it is narrower than the hollow one.
Arrows run from all rows to the never-seen or distinct subset where the shift is visible; deduplication does not always shrink a gap (notdiamond, medical). Gray $n$: distinct prompts in the filled subset. Dashed lines: $\pm 0.10$, the magnitude threshold of the registered criteria.
The orange bracket marks RouteLLM's two co-primary gaps on never-seen prompts ($-0.193$ and $-0.186$). On those prompts, notdiamond's self-harm interval includes zero.
(b)~WildChat medical $\Delta_{\mathrm{adj}}$ against the fraction $q$ of comparison requests sent to the strong model (19 thresholds; $q = 0.2$, $0.5$, and $0.8$ were fixed in the protocol, $q = 0.5$ as the primary point; the other 16 are post hoc). The two hollow RouteLLM points ($q = 0.05$, $0.10$) are where medical requests almost never reach the strong model ($P(\mathrm{strong}) = 0.001$ and $0.007$).
Only $q = 0.5$ carries an interval (95\% user-cluster bootstrap); no CIs exist elsewhere on the curve. RouteLLM's gap shrinks toward both ends, but its per-request privacy loss $\varepsilon$ does not ($1.18$ at $q = 0.2$ against $1.00$ at $0.5$); notdiamond's gap reaches $+0.01$ at $q = 0.05$.
On LMSYS at the 20\% and 80\% operating points (no CIs), RouteLLM keeps its sign in all four categories, while notdiamond turns positive at 80\% in all four (harassment $+0.116$, self-harm $+0.043$, sexual $+0.026$, violence $+0.057$).}
\label{fig:perreq}
\end{figure*}

\subsection{The shift's direction depends on the category and the router}
At the pre-registered 50\% operating point, RouteLLM sends harassment and self-harm requests to the weak model more often than
length-matched comparison requests (Fig.~\ref{fig:perreq}a). On all held-out rows the gaps are $-0.316$ and $-0.238$; both co-primary endpoints of
the replication are met. Medical requests on WildChat, with exploratory LLM labels, show the largest gap, $-0.331$. Sexual and violent
requests, both secondary endpoints, show the opposite direction: they reach the strong model \emph{more} often ($+0.102$ and $+0.071$);
for sexual requests the protocol had predicted a gap below 0.10 in magnitude, which the estimate narrowly exceeds. The second
cost/quality router, notdiamond, sends every category to the weak model more often: it agrees in sign for harassment, self-harm, and
medical requests but disagrees for sexual and violent content ($-0.054$ for both), so the direction is a property of the
category-router pair, not of sensitivity as such. The gap depends on the budget (Fig.~\ref{fig:perreq}b): on WildChat it peaks at $-0.35$ when 55--60\% of
comparison requests reach the strong model, and shrinks toward the extremes ($-0.14$ at 20\%, $-0.13$ at 95\%).

\subsection{Repeated prompts inflate the replication}
By the fingerprint proxy of Section~\ref{sec:estimation}, 76\% of held-out harassment rows belong to a prompt that occurs more than once, and 61\% repeat a prompt
that already appears in the exploration files, so the pre-registered replication is less independent than it appears. On prompts that
never occur in the exploration files, a post hoc subset, the gaps shrink (the arrows in Fig.~\ref{fig:perreq}a) to $-0.193$ for harassment and $-0.186$ for self-harm. Both still meet the
pre-registered criterion (a gap of at least 0.10 whose 97.5\% interval excludes zero). For notdiamond, the harassment gap also survives
($-0.166$), but the self-harm gap, already short of 0.10 on all rows ($-0.097$), falls to $-0.028$ with an interval that includes zero.
WildChat medical requests are 88\% distinct, and their gap is essentially unchanged on
distinct prompts ($-0.312$ against $-0.331$). We treat the never-seen estimates as the robust ones and report the pre-registered estimates alongside.

\subsection{What routing adds to token counts}
On never-seen prompts, RouteLLM's per-request privacy loss at matched length (\ref{eq:eps}) is about 0.4 for harassment, self-harm, and sexual
requests, and about 1.0 for medical requests on WildChat (0.97 on distinct prompts). What the decision adds to token counts
($\Delta$AUC, 20-bin length model) is a different quantity. For self-harm the gain is moderate ($+0.05$ to $+0.06$) and for medical
requests on WildChat it is largest ($+0.11$). For harassment and sexual requests it is about $+0.01$, while length alone already
separates them (AUC 0.78 and 0.75 on distinct prompts). This is consistent with the ceiling effect illustrated by
Remark~\ref{rem:ceiling}, but does not establish that explanation or exhaust length-dependent routing interactions.
The frozen analyses of WildChat and of the LMSYS exploration files used a
log-linear length model, which misses non-monotone length profiles such as that of harassment. It gave a $+0.18$ gain for harassment,
which we withdraw; Table~\ref{tab:perreq} shows post hoc 20-bin values for those rows and lists the frozen ones. Based on that figure,
the held-out replication had predicted (secondary) a 20-bin gain of at least $+0.10$ for harassment; it measured $+0.014$, so the
prediction failed (self-harm: $\ge +0.05$ predicted, $+0.061$ measured).

\begin{figure}[!t]
\centering
\includegraphics[width=\columnwidth]{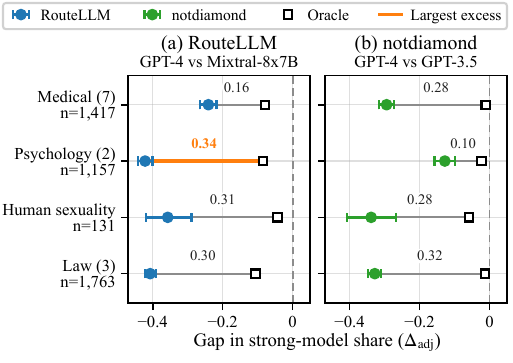}
\caption{Both routers differentiate sensitive MMLU subjects far more than an oracle that routes by realized accuracy gain at the same budget (RouterBench, MMLU subjects). Filled dots: each router's length-matched gap $\dadj$ of (\ref{eq:dadj}) in strong-model share between the group and the non-sensitive prompts (those in none of the four groups), with 95\% percentile-bootstrap CIs over prompts (2,000 resamples); these gaps are pre-registered (E-8b measurement M2), not a primary endpoint. Hollow squares: an oracle with the router's budget that sends the prompts with the largest realized accuracy gain to the strong model, breaking ties at random. The oracle, and hence the excess, is post hoc and uses realized per-prompt correctness; it is one router that maximizes accuracy at that budget, not a deployable router. Its gaps come from a single random tie-break and have no CI. Numbers give the absolute excess $|\dadj(\text{router})-\dadj(\text{oracle})|$, 0.10--0.34; orange marks the largest (RouteLLM, psychology: $-0.422$ vs.\ $-0.085$). The routers' gaps are 13--42 points and the oracle's 1--11 points; we compare them in points, not as ratios, because the oracle's smallest gaps ($-0.011$ and $-0.012$) are close to zero. Against Mixtral-8x7B, part of the oracle's gap reflects accuracy: the strong model's mean accuracy gain on these subjects is 0.11, 0.12, 0.05, and 0.08 lower than on non-sensitive prompts (exploratory; not plotted). Parentheses give the number of MMLU subjects per group; human sexuality ($n=131$) has the widest CIs. Groups are academic subjects, not users' own questions. Sources: \texttt{E-8b/results.json}, \texttt{E-8b/posthoc\_e8b.json}.}
\label{fig:excess}
\end{figure}

\subsection{The routers' gaps far exceed an accuracy oracle's}
Real traffic has no ground truth for whether a request needed the strong model; RouterBench does. On its MMLU subjects, both routers
send medical, psychology, human-sexuality, and law questions to the weak model 13--42 points more often than other prompts
(pre-registered E-8b measurement M2; Fig.~\ref{fig:excess}, Table~\ref{tab:rbleak}). Part of this follows from accuracy: against
Mixtral-8x7B the strong model's realized advantage is smaller on these subjects, by 0.05--0.12 accuracy (gain gap, exploratory), so an
oracle that routes by realized accuracy gain at the same budget (post hoc) also sends them to the weak model more
often. Its gaps are only 4--11 points against Mixtral-8x7B and 1--6 points against GPT-3.5, however, and the routers' gaps exceed them by
10--34 points (post hoc): most of the observed differentiation is not explained by accuracy gain.

\section{User-Level Inference}
\label{sec:userlevel}

This section asks what a user's sequence of routing decisions reveals, using WildChat medical requests, the only sensitive category
whose public data carry a user identifier (Fig.~\ref{fig:userlevel}).

\begin{figure}[!t]
\centering
\includegraphics[width=\columnwidth]{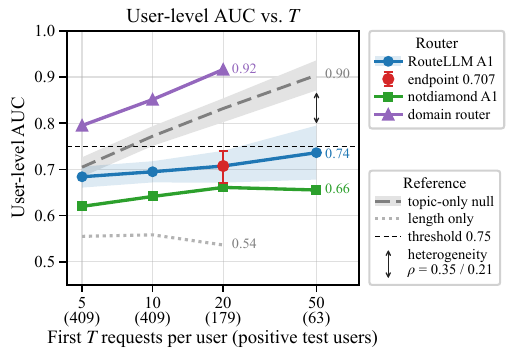}
\caption{User-level AUC rises modestly over the evaluated horizons and remains below the topic-only null; the pre-registered endpoint is missed. With a user's first 20 requests, the bill-only attacker A1 (\ref{eq:a1}) on RouteLLM reaches AUC 0.707 (red; 95\% CI [0.671, 0.740]), below the pre-registered $T{=}20$ threshold of 0.75 (black dashed line). Data: WildChat-1M users with at least 10 requests, on the pre-registered E-5 user split; positives have at least 20\% medical requests (LLM-judge labels), negatives none. The numbers under the ticks count positive test users (negatives: 6,471, 6,471, 2,935, 857). The blue band and the red bar are 95\% percentile-bootstrap CIs over test users (1,000 resamples, stratified by label). The gray dashed line and light gray band show the topic-only null, a model prediction rather than data: the mean and the 2.5--97.5th percentiles of 200 simulations in which each request's strong probability depends only on its topic class. The arrow at $T{=}50$ spans the gap between the two bands, consistent with between-user heterogeneity (Remark~\ref{rem:saturation}). Eligibility changes with $T$, so the curve does not establish an asymptotic ceiling. Heterogeneity is measured as the intra-class correlation $\rho$ of each user's routing decisions: 0.35 among 13,022 negative users and 0.21 among 818 positive users. The domain router (Vela; an upper estimate, as its labels helped select what the judges read) was evaluated only for $T\le 20$. Length only is drawn for $T\le 20$ because at $T{=}50$, with only 63 positives, it falls to 0.47. A random-score control gives 0.49--0.51. The RouteLLM point at $T{=}20$ is the pre-registered primary endpoint (P2). The other $T$ values, notdiamond, the domain router, length only, the random control, the null and $\rho$ are pre-registered secondary analyses. As the LLM-judge label-validity gate failed, every medical result here is exploratory. Table~\ref{tab:userlevel} lists the plotted values.}
\label{fig:userlevel}
\end{figure}

\subsection{Accumulation over the evaluated horizons}
With the first 20 records of a user, the bill-only attacker on RouteLLM reaches AUC 0.707 [0.671, 0.740] on the pre-registered split. This pre-registered primary endpoint
(P2) misses its 0.75 threshold and is exploratory, as the medical label gate failed. Length alone gives 0.54 on the same users and count plus length 0.72, so the routing decisions add about 0.18 AUC beyond token counts. The second cost/quality router gives 0.66. The domain router, whose decision is the topic, gives 0.92, an upper estimate, because its labels helped select the candidates the judges read. In the LMSYS exploration files, it also labels 27\% of the first turns flagged as sexual and 24\% of those flagged as harassment as \emph{health}, against 3\% of unflagged, non-medical first turns. %
On the defense-study split, the bill-only attacker reaches 0.718 and the itemized-log attacker 0.739
(Table~\ref{tab:userlevel}). AUC grows slowly with $T$. The topic-only null predicts 0.83 at $T{=}20$ and 0.90 at
$T{=}50$, against 0.707 and 0.736 observed. Routing rates vary strongly between users within each label class ($\rho = 0.35$ among
negative users and 0.21 among positive ones, over all of their requests), consistent with the mechanism of
Remark~\ref{rem:saturation}. Eligibility changes with $T$, so the curve neither isolates within-cohort accumulation nor establishes
an asymptotic ceiling. The result does not depend on overlapping label and attack windows: labeling users from their
later requests and attacking their first 20 gives 0.69 [0.64, 0.73]. This is the pre-registered stable-trait co-primary endpoint, and,
like E-4's 0.743, it misses 0.75 (Table~\ref{tab:prereg}).

\subsection{An aggregate bill carries most of the signal}
An observer who sees only an aggregate bill, i.e.\ the strong count over a user's first 20 requests, reaches 0.707. A stylized total
cost, which prices each request's tokens 20 or 50 times higher on the strong model and ignores output tokens
(Section~\ref{sec:attackers}), gives 0.65 at either ratio, and total tokens give 0.57. On the defense-study split, itemized logs add about 0.02 over the count. Hiding per-request records while keeping
aggregate bills therefore offers little protection.

\section{Defenses}
\label{sec:defense}

This section evaluates the defenses of Section~\ref{sec:defdef} on WildChat in the router's own units (Fig.~\ref{fig:defense},
Table~\ref{tab:defwild}) and, for the per-request ones, in real accuracy and dollars on RouterBench (Table~\ref{tab:defrb}).

\begin{figure*}[!t]
\centering
\includegraphics[width=\textwidth]{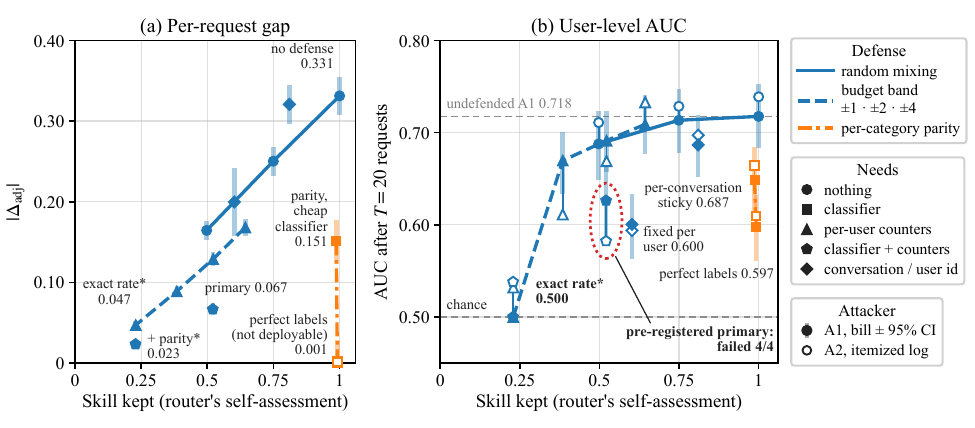}
\caption{Per-request gaps are cheap to close when labels are accurate; user-level inference is not. Each marker is RouteLLM on WildChat
(judge-labeled medical vs.\ other requests, 50\% operating point) with no defense or one post-processing defense; per-category parity
on its own is orange and all other defenses blue. Lines join no defense and random mixing ($\eta=0.25, 0.5$; solid), the per-user budget bands
($\pm4\to\pm2\to\pm1\to$ exact rate; dashed), and parity (cheap classifier $\to$ perfect labels; dash-dot). Marker shape gives what
deployment needs, as in Table~\ref{tab:defwild}; gray dashed lines in (b) mark chance and the undefended A1. The $x$ axis, skill kept,
is the router's own expected benefit relative to random routing at the same budget (0 = random, 1 = undefended). It is the router's
self-assessment, not measured quality: on RouterBench, undefended RouteLLM is 1.30 accuracy points worse than random routing at the
same budget (Section~\ref{sec:realunits}).
(a)~Per-category parity with perfect labels removes the per-request gap ($|\Delta_{\mathrm{adj}}| = 0.001$) at a skill of 0.99, but
it is not deployable; with the cheap classifier the gap stays at 0.151. Whiskers are 95\% user-cluster bootstrap CIs (1,000
resamples) shown as absolute values; the perfect-label interval spans zero, so its whisker runs from 0 to 0.028.
(b)~User-level AUC after $T{=}20$ requests for the bill-only attacker A1 (filled; whiskers are 95\% percentile bootstrap CIs over test users,
1,000 resamples) and the original position-agnostic itemized-log attacker A2 (hollow; CI not drawn); a solid segment joins each defense's A1 and A2. The defense-study
split has 185 positive and 2,891 negative test users; there the undefended A1 is 0.718, against 0.707 on the E-5 split. Only the
exact per-user rate, with or without parity, brings A1 to chance (0.500, degenerate CI), at a skill of 0.23, and only because $T$ is even
(at $T{=}5$: 0.63 and 0.59). Itemized logs still accumulate under it: A2 is 0.53--0.54 at $T{=}20$ and rises to 0.60--0.62 at $T{=}50$
(Table~\ref{tab:defwild}). A post hoc odd-position attacker, not plotted here, reaches 0.730 under the exact rate and 0.667 with
cheap-classifier parity first at $T{=}20$ (Table~\ref{tab:oddattack}).
Status: the defense circled in red in (b), cheap-classifier parity plus a $\pm2$ band, is the pre-registered primary endpoint and failed all
four of its criteria; the undefended baseline and the other nine defenses are pre-registered secondary analyses; the exact-rate
variants ($^\ast$) and the skill-kept measure are post hoc. As for every medical result, the labels come from an LLM judge that failed
its pre-registered precision gate, so all results shown are exploratory (Section~\ref{sec:limits}). Data: E-8 \texttt{results.json} and \texttt{posthoc\_e8.json}.}
\label{fig:defense}
\end{figure*}

\subsection{Per-request gaps close cheaply; user-level inference persists}
Mixing is a poor trade. At $\eta=0.5$ it halves the per-request gap and halves the skill kept, while the user-level AUC falls only from
0.72 to 0.69. Per-category parity with perfect labels removes the per-request gap ($-0.001$; Proposition~\ref{prop:parity}
guarantees zero only against $O\cup U$ on global deciles) and keeps 0.99 of the router's skill. It lowers user-level inference to 0.60 but not to chance. Two explanations remain open: medical
users may be routed differently on their other requests too, or requests the judges missed may stay in the comparison set and keep being
routed to the weak model. With a cheap classifier, the candidate flag of Section~\ref{sec:labels}, parity removes only half of
the gap ($-0.151$), and even this is optimistic: the flag mixes many non-medical requests into the sensitive part, and it also chose
the conversations the judges labeled; elsewhere, only a 2\% sample of requests was labeled. The pre-registered primary
defense, parity with the cheap classifier followed by a $\pm2$ budget band, fails all four criteria: per-request gap, user-level
AUC, retained benefit (0.33 against 0.90), and budget.

\subsection{Budget bands leak through their boundaries}
\label{sec:band}
The two-sided band leaks as Remark~\ref{rem:band} predicts. With slack
$\pm1$, after 20 requests 60.5\% of frequent medical askers sit at the lower edge of the band, against 34.4\% of users with no medical
requests, and 12.4\% against 40.4\% sit at the upper edge (a post hoc count). The count alone then still identifies them (A1 0.67--0.71
across slacks). Only the exact rate, a post hoc variant, hides the strong count, and only at an even number of requests (A1 0.500 at
$T{=}10$, 20, and 50; 0.63 at $T{=}5$, where one undefended decision shows); it does so at a skill of 0.23. The
original position-agnostic itemized-log attacker A2 reaches 0.53 at $T{=}20$ and 0.62 at $T{=}50$. These are attack-specific results, not a ceiling:
the complete log preserves the odd-indexed input decisions (Remark~\ref{rem:exact}), which a position-aware observer can use.
Reusing a conversation's first
decision, an option in some routers, lowers user-level inference only from 0.72 to 0.69 and raises the budget by 2.8 points. Fixing one
model per user lowers it to 0.60, at the cost of discarding content-based routing for returning users.

\subsection{Odd positions retain user-level signal}
\label{sec:oddattack}
A post hoc follow-up tests the transcript identity directly: A2-odd fits on training users' odd positions and scores odd positions
within held-out defended prefixes (Section~\ref{sec:attackers}). At $T{=}20$, it uses ten routing decisions. Under RouteLLM's exact rate, AUC rises from original
A2's 0.532 [0.492, 0.573] to 0.730 [0.700, 0.763]; the paired difference is $+0.199$ [0.148, 0.252], on the same 185 positive
and 2,891 negative test users. With cheap-classifier parity before the exact rate, it reaches 0.667 [0.634, 0.700], a difference
of $+0.129$ [0.076, 0.182]. The corresponding notdiamond AUCs are 0.703 [0.663, 0.742] and 0.672 [0.634, 0.710], with differences
of $+0.146$ [0.095, 0.201] and $+0.096$ [0.040, 0.152]. Table~\ref{tab:oddattack} reports every router--defense--horizon cell,
including the two notdiamond differences at $T{=}50$ whose intervals include zero.

The original A2 AUCs and test counts reproduce in all 16 cells, as do its four $T{=}20$ intervals; the improvement comes from a
different observer, not a correction to the frozen results. The pairwise transcript identities hold for all 982,067 requests from
97,742 users in all four configurations. At even $T$, the strong count still has AUC 0.500. Hence hiding this aggregate does not
hide the ordered log. These are achievable attack performances, not upper bounds on inference. The follow-up is exploratory,
inherits the failed medical-label gate, and uses unadjusted intervals with fixed training fits; users eligible at different $T$
need not be the same. Its defense-study split also differs from the original E-5 split, so 0.730 and E-5's 0.707 are not a paired comparison.

\begin{table*}[!t]
\centering
\caption{Routing policies in real units on RouterBench (36,497 prompts with per-model correctness and dollar cost). Accuracy: mean score of the
selected model. vs.\ rand.: accuracy minus random routing at the same budget. Max leak: largest $|\Delta_{\mathrm{adj}}|$ over the four sensitive groups
of Fig.~\ref{fig:excess}.}
\label{tab:defrb}
\small
\setlength{\tabcolsep}{4.5pt}
\begin{tabular}{@{}lrrrrrrrr@{}}
\toprule
 & \multicolumn{4}{c}{RouteLLM (GPT-4 vs.\ Mixtral-8x7B)} & \multicolumn{4}{c}{notdiamond (GPT-4 vs.\ GPT-3.5)} \\
\cmidrule(lr){2-5}\cmidrule(l){6-9}
Policy & Acc.\ (\%) & vs.\ rand.\ (pp) & \$/1k & Max leak & Acc.\ (\%) & vs.\ rand.\ (pp) & \$/1k & Max leak \\
\midrule
Always weak & 54.7 & $+0.00$ & 0.13 & 0.000 & 61.9 & $+0.00$ & 0.24 & 0.000 \\
Always strong & 78.1 & $+0.00$ & 3.29 & 0.000 & 78.1 & $+0.00$ & 3.29 & 0.000 \\
Router (no defense) & 64.0 & $-1.30$ & 2.12 & 0.422 & 70.6 & $+1.09$ & 1.40 & 0.338 \\
Random mixing, $\eta=0.5$ & 65.0 & $-0.77$ & 1.89 & 0.223 & 70.2 & $+0.50$ & 1.57 & 0.171 \\
Pooled parity (pre-registered) & 63.8 & $-1.52$ & 2.04 & 0.137 & 70.6 & $+1.16$ & 1.36 & 0.166 \\
Per-category parity$^\ast$ & 63.8 & $-1.50$ & 2.04 & 0.064 & 70.6 & $+1.16$ & 1.36 & 0.021 \\
Oracle router$^\ast$ & 82.1 & $+16.79$ & 1.71 & 0.106 & 83.3 & $+13.82$ & 1.68 & 0.059 \\
Oracle + per-category parity$^\ast$ & 82.1 & $+16.78$ & 1.69 & 0.031 & 83.3 & $+13.81$ & 1.67 & 0.010 \\
\bottomrule
\end{tabular}

\smallskip
\footnotesize $^\ast$Post hoc. The oracle routes by realized accuracy gain (ties broken at random); it is an upper bound on accuracy at the
router's budget, not a deployable router.
\end{table*}

\subsection{Real accuracy and dollars}
\label{sec:realunits}
RouterBench tests the router's self-assessment. RouteLLM is 1.30 points \emph{worse} than random routing at the same budget, and its
scores are not ordered by the strong model's real advantage (Spearman $-0.30$ across score deciles). The notdiamond router is 1.09
points better than random. Skill kept in Fig.~\ref{fig:defense} is therefore the router's own belief, and for RouteLLM a defense forfeits little real
quality because there is little to forfeit. The informative comparison is the post hoc oracle. Against Mixtral-8x7B, per-category parity leaves
its accuracy unchanged (82.1\%) and lowers its cost by 0.9\%, while the largest remaining gap falls from 0.106 to 0.031 (against
GPT-3.5: 83.3\%, cost $-0.1\%$, gap 0.059 to 0.010). For the actual routers,
parity changes accuracy by at most 0.2 points and lowers cost by 2--4\%.

Pooled parity is not enough. The pre-registered variant treats the four sensitive subjects as one group. On RouterBench all four start
with negative gaps, but of different magnitudes (RouteLLM: $-0.24$ to $-0.42$). In this evaluation, pooled parity over-corrects
medical questions from $-0.24$ to $+0.14$ while psychology keeps $-0.14$ (Remark~\ref{rem:pooling}), so the leakage criterion fails
although the accuracy ($-0.22$ points) and cost ($-3.8\%$) criteria pass. On LMSYS, contrary to a pre-registered secondary
prediction, RouteLLM shows the same failure when one parity group is formed from any moderation flag, and there the directions differ as well.
Sexual requests, routed to the strong model, and harassment requests, routed to the weak one, offset each other in the pooled group,
so the harassment gap stays at $-0.09$ (from $-0.32$; $+0.01$ with parity for harassment alone) and self-harm barely moves, from
$-0.24$ to $-0.22$. For notdiamond, whose four LMSYS gaps share one sign, the same pooled parity leaves every gap within 0.041, as
predicted. Parity must be applied per category because categories differ in the size of their gaps (RouterBench) or in
their direction (LMSYS).

User-level costs stay self-assessed: WildChat's threshold (0.445) and 80\% of its requests fall below
RouterBench's first score decile (0.555), so the calibration does not transfer.

\section{Discussion}
\label{sec:discussion}

This section turns the measurements into recommendations for three audiences and states what the pre-registration record shows.

\subsection{Recommendations}
\emph{Router developers} can use length-matched parity for a single category or disjoint category partition, with validated labels and calibrated cutoffs. With accurate
labels it equalizes strong rates within every global length decile (Proposition~\ref{prop:parity}), and it cut the WildChat medical gap,
measured on the category's own deciles, to $-0.001$; in a post hoc RouterBench analysis it changed accuracy by at most 0.2 points,
lowered cost by 2--4\%, and left a largest gap of 0.064 (RouteLLM) and 0.021 (notdiamond).
This advice holds only with accurate labels: with a cheap candidate flag, parity moves the medical $\dadj$ only from $-0.331$ to
$-0.151$, and the pre-registered deployable defense built on that flag met none of its four criteria. %
Developers should test every category after any pooling, because categories shift by different amounts and, for RouteLLM on LMSYS,
in opposite directions (Remark~\ref{rem:pooling}; Section~\ref{sec:realunits}). They should also report the operating point, since on WildChat the gap is
$-0.35$ at its peak but $-0.14$ at a 20\% strong-call fraction.

\emph{Gateway operators} whose systems retain routing metadata should document that, when routing depends on content, disabling content
logging does not by itself hide topics, because the retained record may keep the model together with token counts and a caller identifier.
Neither per-conversation stickiness nor the two-sided per-user budget band we test is a privacy control: stickiness lowers user-level
AUC only from 0.72 to 0.69, and under a $\pm1$ band 60.5\% of frequent medical askers end at the lower edge against 34.4\% of users with
no medical requests (Section~\ref{sec:band}). We did not test the per-key spend limits that some gateways offer.
A finite default retention would bound what a later breach or subpoena can reach (Section~\ref{sec:threat}). %

\emph{Adopters} for whom user-level inference matters have no complete tested option. A post hoc exact per-user rate hides the
strong count over an even number of requests (A1 0.500) but not over an odd number (0.63 at $T{=}5$); we did not test the
total-cost bill under it. Its ordered log preserves odd-position decisions, and a post hoc attacker using them reaches AUC 0.730
at $T{=}20$ and 0.793 at $T{=}50$ on RouteLLM (exploratory; Table~\ref{tab:oddattack}). %
It also gives up most of the router's self-assessed value, which we could not express in accuracy or dollars: the router keeps
0.23 of its skill, and the 5\% of users at each end lose or gain at least 47 points of strong-model share. %
Fixing one model per user lowers A1 to 0.60 and A2 to 0.59 while keeping 0.60 of the skill, but it gives up content-based routing
for returning users. Without a defense, hiding itemized logs changes little (A1 0.718 against A2 0.739). %

\subsection{What the pre-registration record shows}
Table~\ref{tab:prereg} reports every pre-registered primary and co-primary endpoint of E-3 to E-9, including failures; secondary
predictions and the decision rules of the two pilots (E-1, E-2) are in the study records. The user-level
primary endpoint was missed (0.707 against 0.75), an earlier marginal pass proved to be the upper end of noise, and the co-primary endpoint
with labels from later requests was missed in both E-4 and E-5; the E-6a sexual-content endpoint missed its threshold, with the sign
opposite to the prediction; both judge-label gates failed, so medical results stay exploratory; and no pre-registered defense endpoint
was met: parity with a budget band failed all four criteria, and pooled parity failed its leakage criterion on RouterBench. Among the
failed secondary predictions, the held-out harassment gain over a 20-bin length model was $+0.014$ against at least $+0.10$ predicted,
and parity over any moderation flag left RouteLLM's LMSYS harassment and self-harm gaps at $-0.094$ and $-0.218$ rather than
within 0.05.
Pre-registration also did not guarantee an independent replication: both of RouteLLM's held-out endpoints were met, but 61\% of the
held-out harassment rows repeated an exploration prompt. %
On never-seen prompts (post hoc), its gaps fall to $-0.193$ and $-0.186$ and still pass; we treat these as the robust estimates. %

\section{Ethics and Responsible Disclosure}
\label{sec:ethics}
This section states how we obtained and handled the data, including four departures from our own handling rules, and how we disclose the findings.

\emph{Data.} We use public research datasets. WildChat-1M~\cite{zhao2024wildchat} is now released under
ODC-BY~\cite{wildchatcard}; its public version excludes conversations flagged as toxic by the OpenAI
Moderations API or Detoxify and conversations that Mireshghallah et al. flagged
as containing personal or sensitive information, and its text was de-identified
with Microsoft Presidio and hand-written rules~\cite{wildchatcard}.
We accepted the LMSYS-Chat-1M Dataset License Agreement~\cite{lmsyscard}, which permits research and commercial use but forbids
transfer to third parties, attempts to identify individuals, and inferring any sensitive personal data encompassed in the dataset.
It does not say whether classifying a message's topic is such inference; we do not claim that our use is covered and asked the
licensor to clarify on October~6, 2026. LMSYS has no user identifier; we make no user-level inference on it, use its moderation flags
and our judges' medical labels only to group requests, and report only aggregate per-category statistics of routing decisions.
RouterBench~\cite{hu2024routerbench} is publicly available, but we have not confirmed its license, so we redistribute none of its
content; its prompt text did not leave the host.
We release no conversation text beyond a few short WildChat fragments quoted in the study records, and no per-request tables,
because token counts and router scores can link a row back to its source conversation (Artifact Statement).
The study involved no contact with the people whose conversations the datasets contain and no attempt to identify them. The user-level
analysis infers a topic attribute of pseudonymous WildChat users and is reported only as aggregate AUCs. The per-request tables kept
from E-4 and E-5 replace hashed IP addresses with random user numbers, which does not prevent that linkage. %
This study was not reviewed by an institutional review board. It is a secondary analysis of previously released public datasets. Each
protocol fixed its handling rules before its data were downloaded; the decision not to release the de-identified per-request tables
came later, after a code audit showed that their rows can be linked back (the E-4 protocol had listed a reproducibility package among
the uses of its table). %

\emph{Handling.} Except in one timing benchmark and the label checks described below, every model that read conversation text
(routers, classifiers, LLM judges, and the machine translation used as a reading aid in E-6a) ran with open weights on a GPU virtual
machine, the host, and no external API scored or translated any conversation. The protocols allowed only aggregate results, text-free
score tables, and gold-set labels to be copied off the host, and assigned the label spot checks and gold sets to the author. The working
directories that held conversation text in E-5, E-6a, and E-7-lite were deleted when each study ended.
The earlier studies (E-2 to E-5) departed from these rules four times.
First, to estimate E-2's run time on the author's computer, its two WildChat files, which hold message text and hashed IP addresses,
were downloaded there; the router and the topic classifier scored 640 English first turns, and the script printed only timing and
throughput summaries. The directory was deleted about nine minutes after the download began and about six minutes after its
completion was observed; the study ran on the host, and its deviation log (v1.2) recorded the timing but not the download.
Second, the label spot-check snippets of E-2 (40), E-3 (20), and E-4 (600), at most 200 characters each, and the 400 E-5 gold-set
lines were printed in the terminal of the AI coding assistant used in this study (Claude, by Anthropic), so their text was sent to the
assistant's provider, against our rule that text stays on the host. Each E-5 line held the item number, the sampling stratum, and at
most 120 characters of the previous message and 300 of the message, capped at 330 characters (about 170--180 of the message after a
long previous one; the judges read up to 1,000 and 300). With no input from the author, the assistant judged the spot checks and
assigned the 400 E-5 reference labels (Section~\ref{sec:labels}), not blind: the lines showed classifier scores and sampling pools
(E-2, E-4) or the stratum, which encodes the judge's label (E-5), the E-3 sample held only classifier positives, and the
assistant knew the hypotheses, so agreement may be overestimated. The E-4 and E-5 records called this the author's work; corrections appended to the E-2
to E-5 deviation logs on October~6, 2026 record who read and labeled each set.
Third, when the E-5 analysis was rerun after a merge-key fix, score tables containing WildChat's hashed IP addresses and conversation
hashes were copied to a temporary directory on the author's local computer and deleted after the analysis, within about an hour; the
judge outputs copied with them had their raw-text column removed on the host first.
Fourth, a file of gold-set keys that included conversation hashes was committed to the study's repository. We rewrote the history to
purge it from every commit; commits before the E-5 results commit, including the earlier protocol freezes, kept their identifiers
(Artifact Statement).
Hashed IP addresses and conversation hashes are fields of the public WildChat release, but our protocols kept them on the host. After a
code audit found that the E-4 and E-5 snippets had passed through the assistant, we moved all reading of conversation text off it. The
only such reading afterwards was of the 150-item E-6a gold set, which the author read and labeled in the author's own browser, on a
page served only on the host's loopback interface: the snippets and their machine translations traveled through an encrypted SSH
tunnel to the author's computer (a logged deviation from the protocol's SSH terminal), so the text, not only identifiers and labels,
left the host, but no AI assistant or external API received it. E-7-lite involved no human reading.

\emph{Sensitive content.} No one read any text in the sexual, self-harm, harassment, or violence groups. These groups are defined only
by the moderation flags shipped with LMSYS-Chat-1M; their messages were not sent to the judges, and none entered a gold set.
Conversations flagged for sexual content involving minors were excluded from every analysis and table, and only their number was kept.
The pending probe of OpenRouter's Auto Router~\cite{openrouterauto} will send only public MMLU questions~\cite{hendrycks2021mmlu}
and author-written template questions that differ only in topic: no real conversation, no harassment, self-harm, or suicide content,
and no sexual content beyond MMLU's human-sexuality questions. It attempts no prompt injection or jailbreak and stores no response text. %

\emph{Disclosure.} We notified the maintainers of the three systems we measured, RouteLLM, Not Diamond, and vLLM Semantic Router,
through private channels on October~2, 2026, six days before this preprint was submitted. The notice summarized the results and the
mitigations of Section~\ref{sec:defense} and offered a private draft on request. By October~7, 2026, Not Diamond had replied, with no factual corrections; the other two had not.
We did not measure the routers of OpenRouter and LiteLLM and describe them only from their public documentation; we will send them the
paper when it is posted and incorporate their corrections in a revision.
We do not treat the finding as a software vulnerability: there is no defect to patch in the usual sense, because the leak follows from
content-dependent routing combined with standard metadata logging. The usual vulnerability embargo therefore does not apply; the notice
gave maintainers the chance to correct facts before release. Gateways and cloud platforms that only record the served model are not
recipients, because faithful logging is not a defect.

\section{Limitations}
\label{sec:limits}
This section states what our evidence does not establish.

(1)~\emph{Labels.} The medical judges failed both pre-registered label gates (Section~\ref{sec:labels}), so every medical result is
exploratory. Only the LMSYS gate (precision 0.75, weighted recall 0.69) tested them against human labels, from one annotator, the
author, so we cannot report agreement between humans. The WildChat gate (0.70) compared the judge with reference labels that the AI
coding assistant assigned, not blind, from a shorter window (at most 300 against 1,000 characters; Section~\ref{sec:ethics}), that is,
one model with another; the assistant's labels agree with the label judge and a second LLM judge about as well as the two judges agree
with each other (Cohen's $\kappa = 0.57$ and $0.65$, against $0.64$), and this agreement and the 0.70 may be overestimated.
Recall on WildChat cannot be estimated for the population, and the WildChat comparison set contains medical requests that no judge
saw: on a 2\% sample of non-candidate messages, the judge labeled 1.77\% as medical.
In three simulated misclassification scenarios, RouteLLM's true LMSYS gap keeps its sign but ranges from $-0.10$ to $-0.23$, against
$-0.162$ observed, so its size is uncertain. %
The moderation flags that define the confirmatory categories are themselves classifier output: the per-message OpenAI moderation
outputs shipped with LMSYS-Chat-1M, collected in 2023~\cite{zheng2024lmsys,lmsyscard}. The paper and dataset card do not state which
moderation model/version generated the released per-message outputs; consequently, those labels cannot be reproduced exactly from the
release metadata alone. The current API documentation~\cite{openai_moderation} describes only its current model. We did not validate them, because we chose not to have anyone read
flagged text (Section~\ref{sec:ethics}). %

(2)~\emph{Duplicates.} Public chat logs repeat prompts, so the held-out replication is less independent than its design assumed
(Section~\ref{sec:discussion}). On never-seen prompts, notdiamond's self-harm gap, already short of the pre-registered criterion on all
held-out rows, becomes indistinguishable from zero (Section~\ref{sec:perreq}). %
The fingerprint (both rounded router scores and the token count) is not a text-identity test. Distinct texts can share it, including
texts that differ only beyond both routers' truncation windows and have equal token counts; near-duplicates with different fingerprints
can remain. On LMSYS, the pre-registered intervals resample requests and ignore this dependence; the fingerprint-cluster intervals
of Table~\ref{tab:perreq} account for repeated fingerprints but hold the
comparison set fixed.

(3)~\emph{Routers.} We study two open cost/quality routers, both BERT-class classifiers, and one domain router. Routers with other
architectures may leak differently. For the domain router, we treat the \emph{health} label as the decision and did not map domains to
models under its default configuration. We evaluated the open notdiamond-0001 model, not Not Diamond's hosted service. Not Diamond informed us that notdiamond-0001
has not been actively developed for several years; our results describe that open model, not the company's current routing service.
We have not measured the leakage of a
hosted commercial router (Wang et al.~\cite{wang2026misguided} measure how six evaluated commercial routers, OpenRouter's Auto Router among them,
route by source benchmark, without a privacy analysis); the pre-registered probe of OpenRouter's Auto Router (Section~\ref{sec:ethics}) is released with the artifact,
and its results will be reported in a revision.

(4)~\emph{Users.} User-level evidence covers one category and one dataset: medical requests with exploratory labels, on WildChat, the
only public source we use that has a user identifier. Its pre-registered primary endpoint was not met (AUC 0.707 against 0.75 at $T{=}20$).
User-level inference for harassment, self-harm, and sexual requests is unmeasured. Users are hashed IP addresses, which merge people who
share an address and split people who use several, and the public release excludes conversations flagged as toxic by the OpenAI Moderations API or Detoxify and those that
Mireshghallah et al. flagged as containing personal or sensitive information~\cite{wildchatcard}. Both effects likely understate user-level inference. User labels come from the same
content-based judge labels as the per-request analysis, not from attributes that users reported, and because the domain router's labels
helped select the messages the judges read, its 0.92 is an upper estimate. The attacker is also given labeled
training users, which favors the attacker. The odd-position follow-up is post hoc, assumes a complete ordered prefix with known
starting phase, and inherits the failed medical-label gate. Its intervals hold training fits fixed and are not adjusted for multiple
comparisons; differing eligible cohorts across horizons prevent a within-cohort accumulation claim.

(5)~\emph{Real units.} Accuracy and dollar costs are measured on RouterBench. Its sensitive groups are academic MMLU subjects, not users'
own questions, and its prices are the 2023--2024 prices recorded in the benchmark. %
The oracle routes by realized per-prompt accuracy gain, so it is not deployable, and it is only one quality-optimal router
(Section~\ref{sec:realdef}). The routers' gaps (13--42 points) are pre-registered; the oracle's (1--11), the comparison, and the
real-unit cost of per-category parity are post hoc.
The calibration does not transfer to real chats: 80\% of WildChat requests fall in RouterBench's lowest score decile (below 0.555).
User-level defense costs therefore remain in the router's own units, which for RouteLLM do not track real accuracy (Spearman $-0.30$
across score deciles). %

(6)~\emph{Parity resolution.} Parity is enforced within global length deciles, and Proposition~\ref{prop:parity} equalizes rates only
on these bins. Measured on each category's own deciles, a residual remains, up to 0.064 on RouterBench (RouteLLM, law questions), above
the 0.02 leakage criterion of the pre-registered RouterBench study. With perfect labels on WildChat, parity leaves no gain over a 20-bin
length model ($\Delta\mathrm{AUC} = +0.000$), but we tested it only against the original observers of Section~\ref{sec:attackers},
not the odd-position attacker, which we evaluate only for the two exact-rate variants.
An observer that uses exact token counts or other record fields may recover residual signal, as length-matched controls fail to close
selection channels in agent workflows~\cite{xu2026sicc}.

(7)~\emph{Operating point.} Leakage depends on the threshold. We calibrate it to a strong-call fraction, the procedure RouteLLM
documents~\cite{routellmrepo}. At its package-default threshold with an uncalibrated model pair, RouteLLM's similarity-weighted ranking
(\texttt{sw\_ranking}) router sent every task to the weak model~\cite{kumar2026routers}; such a router leaks little, but it also gives up content-based routing. We therefore report the whole
operating-point range. On WildChat, the medical gap peaks at $-0.35$ when 55--60\% of comparison requests reach the strong model and
shrinks to $-0.14$ at 20\% and $-0.13$ at 95\%. %
On LMSYS, RouteLLM's harassment and self-harm gaps keep their sign at the 20\% and 80\% points; notdiamond's reverse at 80\%
($+0.116$ and $+0.043$). %

(8)~\emph{Scope.} Both datasets are restricted to English conversations. LMSYS contributes only first turns, because it has no user
identifier and the first turn is the setting RouteLLM was trained for, and it was collected in 2023, so usage may have shifted since.
We measure an association at matched length, not its cause. Sensitive requests may also differ in format or phrasing, for example
exam-style versus conversational, which we did not control; register is known to shift complexity-based routing, though mainly
through input length~\cite{koul2026register}, for which our gaps adjust. The observer's inference does not depend on the cause, but a mitigation
that targets topic might.

\section{Conclusion}
\label{sec:conclusion}
This paper treated the routing decision that many gateways and some cloud platforms can record as a privacy channel and measured it
on real traffic. At matched length, the direction of the shift depends on the category and the router. For RouteLLM at the 50\% operating
point, harassment and self-harm requests reach the strong model 19 points less often than comparable requests on prompts never
seen during exploration (post hoc; 24--32 points on all held-out rows, as pre-registered), sexual requests 10 points more often
(a secondary endpoint predicted to stay below 10), and medical requests, with exploratory LLM-judge labels, 31 points less often on
distinct WildChat prompts (post hoc; 33 on all rows). The notdiamond router sends all four less often, but its never-seen self-harm
gap is indistinguishable from zero.
An observer who sees only how many of a user's first 20 requests RouteLLM sent to the strong model separates frequent medical askers
from users with no medical requests at an AUC of 0.71, exploratory and below the pre-registered primary endpoint's 0.75 (0.92, an
upper estimate, for a domain router's health label). %
On RouterBench, the routers' gaps on sensitive subjects (13--42 points, pre-registered) exceed those of an oracle that routes by
realized accuracy gain at the same budget (1--11 points, post hoc), and per-category parity leaves a largest gap of at most 0.064 for
at most 0.2 accuracy points (post hoc). %
Pooled parity fails because categories shift by different amounts and, for RouteLLM on LMSYS, in opposite directions; per-conversation
stickiness barely helps, and budget bands leak at their edges. Only a post hoc exact per-user rate hides the strong count, and only
at even request counts; its ordered log preserves odd-position decisions exactly. A post hoc attack using them reaches AUC 0.73
after 20 RouteLLM requests (exploratory), even as the defense gives up most of the router's self-assessed value, a cost we could state
only in the router's own units. %
Practitioners should therefore treat the model field as content-derived data: consider per-category parity with validated labels and
calibrated cutoffs, and stop counting per-conversation stickiness and budget bands as privacy controls.

The limitations of Section~\ref{sec:limits} set the future work. Medical magnitudes need judge labels validated against human labels
on both datasets; ours failed the LMSYS gate, and the WildChat gate used the AI coding assistant's non-blind labels. User-level measurement of
harassment and self-harm needs user-linked data, such as the full WildChat release, for which access is under review. Hosted commercial
routers need a leakage measurement, starting with the pre-registered probe of OpenRouter's Auto Router. A controlled rewrite of the same requests
would show whether part of each gap reflects format or wording rather than topic. The per-user rate distribution that minimizes
user-level inference at a given budget, and the residual leakage of parity to adaptive observers that use exact token counts, are
untested. %

\appendix[Supporting Tables]
This appendix gives the full numbers behind Figs.~\mbox{\ref{fig:excess}--\ref{fig:defense}}: leakage on RouterBench
(Table~\ref{tab:rbleak}), user-level inference (Table~\ref{tab:userlevel}), and the defenses on WildChat (Table~\ref{tab:defwild}).
Table~\ref{tab:oddattack} reports all 16 configurations of the post hoc odd-position follow-up, including paired differences and uncertainty.

\begin{table*}[!t]
\centering
\caption{Leakage on RouterBench's MMLU subjects (comparison set: all other RouterBench prompts). Routers' $\Delta_{\mathrm{adj}}$:
pre-registered (E-8b measurement M2). Oracle (post hoc, as is any comparison with it): a router with the same budget that sends the prompts with
the largest realized accuracy gain to the strong model, with ties broken at random; it is one router that maximizes accuracy at that budget.
Gain gap (exploratory): the strong model's mean accuracy gain on the group minus its gain on the comparison set.}
\label{tab:rbleak}
\small
\begin{tabular}{@{}lrlrrlr@{}}
\toprule
 & & \multicolumn{3}{c}{RouteLLM (GPT-4 vs.\ Mixtral-8x7B)} & \multicolumn{2}{c}{notdiamond (GPT-4 vs.\ GPT-3.5)} \\
\cmidrule(lr){3-5}\cmidrule(l){6-7}
Group & $n$ & $\Delta_{\mathrm{adj}}$ [95\% CI] & Oracle & Gain gap & $\Delta_{\mathrm{adj}}$ [95\% CI] & Oracle \\
\midrule
Medical (7 subjects) & 1,417 & $-0.241$ [$-0.266$, $-0.217$] & $-0.079$ & $-0.11$ & $-0.294$ [$-0.316$, $-0.272$] & $-0.011$ \\
Psychology (2) & 1,157 & $-0.422$ [$-0.443$, $-0.401$] & $-0.085$ & $-0.12$ & $-0.128$ [$-0.157$, $-0.098$] & $-0.023$ \\
Human sexuality & 131 & $-0.357$ [$-0.419$, $-0.289$] & $-0.043$ & $-0.05$ & $-0.338$ [$-0.408$, $-0.267$] & $-0.059$ \\
Law (3) & 1,763 & $-0.407$ [$-0.422$, $-0.391$] & $-0.106$ & $-0.08$ & $-0.328$ [$-0.347$, $-0.309$] & $-0.012$ \\
\bottomrule
\end{tabular}
\end{table*}

\begin{table*}[!t]
\centering
\caption{User-level inference on WildChat (medical). A test user is positive if $\ge$20\% of their requests are medical and negative if
none are; the attacker sees their first $T$ requests. A1: count of strong-routed requests (an aggregated bill); A1+L: plus mean length;
A2: itemized log, per-request likelihood ratio of the routing bit given the length decile. Topic-only null: the AUC if users differed only
in topic mix. Domain router: an upper estimate, as its labels helped select what the judges read. Pos.: positive test users.}
\label{tab:userlevel}
\small
\setlength{\tabcolsep}{4pt}
\begin{tabular}{@{}rrrrrrrrrrrrrr@{}}
\toprule
 & \multicolumn{6}{c}{Pre-registered E-5 split (Fig.~\ref{fig:userlevel})} & \multicolumn{7}{c}{Defense-study split (Table~\ref{tab:defwild})} \\
\cmidrule(lr){2-7}\cmidrule(l){8-14}
 & & Length & RouteLLM & notdiamond & Domain & Topic-only & & Length & \multicolumn{3}{c}{RouteLLM} & \multicolumn{2}{c}{notdiamond} \\
\cmidrule(lr){10-12}\cmidrule(l){13-14}
$T$ & Pos. & only & A1 & A1 & router & null & Pos. & only & A1 & A1+L & A2 & A1 & A2 \\
\midrule
5 & 409 & 0.555 & 0.684 & 0.620 & 0.795 & 0.705 & 409 & 0.566 & 0.675 & 0.679 & 0.696 & 0.632 & 0.691 \\
10 & 409 & 0.558 & 0.695 & 0.642 & 0.851 & 0.772 & 409 & 0.566 & 0.690 & 0.695 & 0.713 & 0.642 & 0.708 \\
20 & 179 & 0.537 & 0.707 & 0.661 & 0.916 & 0.832 & 185 & 0.560 & 0.718 & 0.731 & 0.739 & 0.680 & 0.736 \\
50 & 63 & 0.473 & 0.736 & 0.656 & -- & 0.904 & 62 & 0.521 & 0.769 & 0.803 & 0.800 & 0.674 & 0.712 \\
\bottomrule
\end{tabular}

\smallskip
\footnotesize Pre-registered primary endpoint (P2): RouteLLM's A1 at $T{=}20$ on the E-5 split, 0.707 [0.671, 0.740], below 0.75; exploratory (medical label gate failed).
\end{table*}

\begin{table*}[!t]
\centering
\caption{Defenses as post-processing of RouteLLM on WildChat (medical). Skill kept (post hoc): the router's own expected benefit relative to random
routing at the same budget (0 = random, 1 = undefended); it is the router's self-assessment, not measured quality (Section~\ref{sec:realunits}).
It replaces the pre-registered retained benefit, the defended router's expected benefit as a fraction of the undefended one's, on which the
primary defense scores 0.33. User-level columns are AUCs of the attackers of Table~\ref{tab:userlevel}.}
\label{tab:defwild}
\small
\begin{tabular}{@{}llrrrrrr@{}}
\toprule
 & & $\Delta$Budget & Skill & Per-request & \multicolumn{3}{c}{User-level AUC} \\
\cmidrule(l){6-8}
Defense & Needs & (pp) & kept & $\Delta_{\mathrm{adj}}$ & A1, $T{=}20$ & A2, $T{=}20$ & A2, $T{=}50$ \\
\midrule
No defense & -- & $+0.0$ & 1.00 & $-0.331$ & 0.718 & 0.739 & 0.800 \\
Random mixing, $\eta=0.25$ & -- & $+0.4$ & 0.75 & $-0.250$ & 0.713 & 0.729 & 0.781 \\
Random mixing, $\eta=0.5$ & -- & $+0.8$ & 0.50 & $-0.164$ & 0.688 & 0.711 & 0.757 \\
Per-category parity, cheap classifier & classifier & $+0.0$ & 0.99 & $-0.151$ & 0.648 & 0.664 & 0.719 \\
Per-category parity, perfect labels & (not deployable) & $+0.0$ & 0.99 & $-0.001$ & 0.597 & 0.610 & 0.642 \\
Per-user budget band $\pm1$ & per-user counters & $+1.8$ & 0.38 & $-0.089$ & 0.670 & 0.611 & 0.672 \\
Per-user budget band $\pm2$ & per-user counters & $+1.5$ & 0.52 & $-0.129$ & 0.691 & 0.669 & 0.698 \\
Per-user budget band $\pm4$ & per-user counters & $+1.2$ & 0.64 & $-0.168$ & 0.709 & 0.733 & 0.751 \\
Parity + band $\pm2$ (primary) & both & $+1.4$ & 0.52 & $-0.067$ & 0.626 & 0.582 & 0.666 \\
Exact per-user rate$^\ast$ & per-user counters & $+1.9$ & 0.23 & $-0.047$ & 0.500 & 0.532 & 0.618 \\
Parity + exact rate$^\ast$ & both & $+1.9$ & 0.23 & $-0.023$ & 0.500 & 0.538 & 0.600 \\
Sticky model per conversation & conversation id & $+2.8$ & 0.81 & $-0.321$ & 0.687 & 0.697 & 0.773 \\
Fixed model per user & user id & $-1.9$ & 0.60 & $-0.200$ & 0.600 & 0.594 & 0.572 \\
\bottomrule
\end{tabular}

\smallskip
\footnotesize $^\ast$Post hoc, added after the pre-registered band variants revealed the boundary leak (Section~\ref{sec:band}).
A1 is 0.500 only at even $T$; at $T{=}5$ it is 0.63 and 0.59.
\end{table*}

\begin{table*}[!t]
\centering
\caption{Odd-position attacks on exact-rate defenses on WildChat medical labels: all configured routers, defenses, and horizons.
This analysis is post hoc and exploratory; the medical-label validity gate failed.
$T$ counts total original requests, and $m=\lceil T/2\rceil$ counts the observed odd-position decisions.
Original A2 fits and scores all positions; A2-odd fits only training users' odd positions and scores those within the same test prefix.
Intervals are unadjusted 95\% percentile CIs from 1,000 label-stratified test-user bootstrap resamples (seed 20260928), with training fits fixed.
The difference is A2-odd minus A2, using paired resamples.}
\label{tab:oddattack}
\small
\setlength{\tabcolsep}{5pt}
\begin{tabular}{@{}rrrccc@{}}
\toprule
$T$ & $m$ & \shortstack{Test users\\positive/negative} & Original A2 [95\% CI] & A2-odd [95\% CI] & Paired difference [95\% CI] \\
\midrule
\multicolumn{6}{@{}l}{\emph{RouteLLM: exact rate}} \\
5  & 3  & 409/6502 & $0.514\ [0.486,0.542]$ & $0.684\ [0.661,0.707]$ & $+0.170\ [0.139,0.202]$ \\
10 & 5  & 409/6502 & $0.502\ [0.471,0.532]$ & $0.702\ [0.678,0.725]$ & $+0.201\ [0.165,0.235]$ \\
20 & 10 & 185/2891 & $0.532\ [0.492,0.573]$ & $0.730\ [0.700,0.763]$ & $+0.199\ [0.148,0.252]$ \\
50 & 25 & 62/831   & $0.618\ [0.547,0.683]$ & $0.793\ [0.743,0.836]$ & $+0.175\ [0.097,0.259]$ \\
\addlinespace
\multicolumn{6}{@{}l}{\emph{RouteLLM: D2-cand + exact rate}} \\
5  & 3  & 409/6502 & $0.497\ [0.470,0.524]$ & $0.617\ [0.590,0.641]$ & $+0.120\ [0.085,0.153]$ \\
10 & 5  & 409/6502 & $0.506\ [0.476,0.537]$ & $0.639\ [0.614,0.663]$ & $+0.134\ [0.096,0.170]$ \\
20 & 10 & 185/2891 & $0.538\ [0.495,0.579]$ & $0.667\ [0.634,0.700]$ & $+0.129\ [0.076,0.182]$ \\
50 & 25 & 62/831   & $0.600\ [0.521,0.668]$ & $0.715\ [0.658,0.770]$ & $+0.115\ [0.032,0.205]$ \\
\addlinespace
\multicolumn{6}{@{}l}{\emph{notdiamond: exact rate}} \\
5  & 3  & 409/6502 & $0.487\ [0.457,0.517]$ & $0.672\ [0.646,0.701]$ & $+0.185\ [0.143,0.229]$ \\
10 & 5  & 409/6502 & $0.535\ [0.506,0.563]$ & $0.690\ [0.666,0.718]$ & $+0.155\ [0.118,0.193]$ \\
20 & 10 & 185/2891 & $0.557\ [0.514,0.595]$ & $0.703\ [0.663,0.742]$ & $+0.146\ [0.095,0.201]$ \\
50 & 25 & 62/831   & $0.622\ [0.552,0.689]$ & $0.686\ [0.620,0.751]$ & $+0.063\ [-0.041,0.166]$ \\
\addlinespace
\multicolumn{6}{@{}l}{\emph{notdiamond: D2-cand + exact rate}} \\
5  & 3  & 409/6502 & $0.488\ [0.458,0.516]$ & $0.636\ [0.607,0.665]$ & $+0.148\ [0.106,0.190]$ \\
10 & 5  & 409/6502 & $0.559\ [0.534,0.585]$ & $0.647\ [0.622,0.675]$ & $+0.088\ [0.051,0.126]$ \\
20 & 10 & 185/2891 & $0.576\ [0.534,0.614]$ & $0.672\ [0.634,0.710]$ & $+0.096\ [0.040,0.152]$ \\
50 & 25 & 62/831   & $0.630\ [0.566,0.704]$ & $0.674\ [0.614,0.736]$ & $+0.043\ [-0.064,0.144]$ \\
\bottomrule
\end{tabular}

\smallskip
\footnotesize The two attacks use identical test users at each $T$; eligibility changes across $T$.
The complete ordered prefix and its starting phase are assumed observable. At even $T$, the strong count remains fixed and its AUC is 0.500.
These attack AUCs are achievable performance, not upper bounds on inference. Both notdiamond difference intervals at $T=50$ include zero.
\end{table*}

\section*{Artifact Statement}
This statement lists what we release and what we withhold.
We release, for every pre-registered study, the protocol, the deviation log where deviations occurred, and the other study records (in
Traditional Chinese); the analysis code, including the script that generates Tables~\ref{tab:prereg}--\ref{tab:defwild} from the
result files; the aggregate results and run logs; the router revisions and weight hashes; the gold-set labels by item number (E-6a:
the author's; E-5: the AI coding assistant's); the aggregate label-error simulation of Limitations~(1), from E-7, whose draft protocol
was never frozen and is not released; the per-question router scores of the MMLU pilot; the E-9 probe, which rebuilds every prompt from a fixed seed, and its analysis
script; and the threat-model evidence table (in Traditional Chinese), with the URL and access date of each of its 59 sources and, for
55, the most recent snapshot that the Wayback availability API returned (for two, of the main-branch or pre-redirect URL).
Commit hashes in the study records refer to our private research repository, which is not public because its history contains
per-request tables; Table~\ref{tab:prereg} gives the freeze dates of E-3 to E-9. A rewrite of that history, recorded in the E-5
deviation log, removed a file with conversation hashes; it changed the identifiers of the E-5 results commit and later commits, but not
their dates. These dates are our own record of when each protocol was committed (E-5: before any judge label existed; E-8: before any
parity or budget defense was applied, after E-5b had measured its two stickiness defenses post hoc; E-9: before any request was
sent), not an external timestamp.
The artifact excludes the text-free per-request tables of the chat datasets and RouterBench, because router scores with token counts
can link a row back to its source and the licenses restrict redistribution (LMSYS-Chat-1M) or are unverified (RouterBench); it also
excludes our internal notes and disclosure correspondence.
The public datasets are on Hugging Face (\code{allenai/WildChat-1M}, \code{lmsys/lmsys-chat-1m} after accepting its license
agreement, \code{withmartian/routerbench}, and \code{cais/mmlu} for the pending probe); the study records give the revisions we used.
The artifact is available at \url{https://github.com/luka-krixvon/routing-metadata-leakage}.

\clearpage
\twocolumn
\bibliographystyle{IEEEtran}
\bibliography{refs}

\end{document}